\documentclass{article}

\usepackage{amssymb,amsmath, amsthm}
\usepackage{url}

\usepackage{setspace}
\usepackage{graphicx} 
\usepackage{color}  
\usepackage{pifont}
\usepackage[mathscr]{euscript}

\usepackage[T1]{fontenc}
\usepackage{epigraph}

\usepackage[lighttt]{lmodern}
\usepackage{caption}
\usepackage{booktabs} 
\usepackage{geometry}

\usepackage{sectsty}   

\usepackage{titlesec}

\newtheorem{definition}{Definition}[section]
\newtheorem{proposition}{Proposition}[section]
\newtheorem{theorem}{Theorem}[section]
\newtheorem{corollary}{Corollary}[section]
\newtheorem{lemma}{Lemma}[section]

\newcommand{\low}{\textit{\textbf{l}}}
\newcommand{\upp}{\textit{\textbf{u}}}

\begin{document}

	\begin{center}
		\doublespacing	{\textbf{\huge{Double Successive Rough Set Approximations}}}	\\
		
		\vspace{0.8cm}
		Alexa Gopaulsingh\\
		
		\vspace{1mm}
	\singlespacing{	Central European University\\
		Budapest, Hungary\\
		\ttfamily{\bfseries{gopaulsingh\_alexa@phd.ceu.edu}}}\\
		\end{center}
		
\vspace{1cm}

\noindent \textbf{Abstract.} We examine double successive approximations on a set, which we  denote by $L_2L_1, \ U_2U_1, \\ U_2L_1,$  $L_2U_1$ where $L_1,  U_1$  and $L_2, U_2$ are based  on generally non-equivalent equivalence relations $E_1$ and $E_2$ respectively, on a finite non-empty set $V.$ We consider the case of these operators being given fully defined on its powerset $\mathscr{P}(V).$ Then, we investigate if we can reconstruct the equivalence relations which they may be based on. Directly related to this, is the question of whether there are unique solutions for a given defined operator and the existence of conditions which may characterise this. We find and prove these characterising conditions that equivalence relation pairs should satisfy in order to generate  unique such operators.
 
\vspace{5mm} \noindent \textbf{Keywords:} Double Approximations $\cdot$ Successive Approximations $\cdot$ Double, Successive Rough Set Approximations

	\section{Successive Approximations}   
	Double successive rough set approximations here, are considered using two, generally different equivalence relations. These are interesting because one can imagine a situation or model where sets/information to be approximated is input through two different approximations before returning the output. It is possible that for example  heuristics in the brain can be modelled using such layered approximations. Decomposing successive approximations  into constituent parts is somewhat analogous to decomposing a wave into sine and cosine parts using Fourier analysis.  
	
	In our case, we have two equivalence relations $E_1$ and $E_2$ on a set $V$ with lower and upper approximations operators acting on its powerset $\mathscr{P}(V),$ denoted by $L_1, \ U_1$ and $L_2, \ U_2$ respectively. What if we knew the results of passing all the elements in $\mathscr{P}(V)$ through $L_1$ and then $L_2,$ which we denote by $L_2L_1.$ Could we then reconstruct $E_1$ and $E_2$ from this information? In this paper, we will investigate this question and consider the four cases of being given a defined $L_2L_1, \ U_2U_1, \  U_2L_1, \ L_2U_1$ operators. We will find that two equivalence relations do not always produce unique such operators but that some pairs do. We find and characterise conditions which the pairs of equivalence relations must satisfy for them to produce a unique operator.  Cattaneo and Ciucci found that preclusive relations are especially useful for  rough approximations in information systems in \cite{PR}. For the $L_2L_1$ case we will show that these conditions from a preclusive relation between pairs of equivalence relations on a set and so we can define a related notion of independence from it. After this, we will find a more conceptual but equivalent version of the conditions of the uniqueness theorem. These conditions are more illuminating in that we can easier see why these conditions work while the conditions in the first version of the theorem are easier to use in practice.  Lastly, we will consider the cases of the remaining operators, $U_2U_1, \ U_2L_1$ and $L_2U_1.$ We note that the $L_2L_1$ and $U_2U_1$ cases are dual to each other and similarly for the $U_2L_1$ and  $L_2U_1$ cases.

	Rough set theory has quite a large number of practical applications. This is due in part to the computation of reducts and decision rules for databases. Predictions can be made after the data is mined to extract decision rules of manageable size (i.e. attribute reduction). In this way,   rough set theory can be used to make decisions using data in the absence of major prior assumptions as argued in more detail in \cite{Baye}.  Hence in retrospect, it is perhaps not so surprising that this leads to tremendous applications. Therefore, rough set analysis is added to the tools, which includes regression analysis and  Bayes' Theorem, for pattern recognition and feature selection in data mining, see  \cite{DM4, DM1, DM2, DM3, DM5, DM6, DM7, DM8}.  The resulting applications include in medical databases \cite{MD1, MD2, MD3, MD4, MD5, MD6, MD7}, cognitive science  \cite{CG1, CG2, CG3, CG4, CG5}, artificial intelligence and machine learning \cite{AC1, AC2, AC3, AC4, AC5, AC6, AC7} and  engineering  \cite{EN1, EN2, EN3, EN4, EN5}. Indeed in \cite{TSid}, Yao noted  that there is currently an imbalance in the literature between the conceptual unfolding of rough set theory and its practical computational progress. He observed that the amount of computational literature currently far exceeds the amount of conceptual, theoretical literature. Moreover, he made the case that the field would prosper from a correction of this imbalance.  To illustrate this, he began his recommendation in \cite{TSid} by formulating a conceptual example of reducts that unifies three reduct definitions used in the literature which on the surface look different. We strongly agree that more efforts to find  conceptual formulations of notions and results would increase the discovery of unifying notions.  This would greatly aid the aim of making a cohesive and coherent map of the present mass of existing literature. In this direction, we have developed subsections 4.2 and 4.3 in section 4.

		\section{Basic Concepts of Rough Set Theory}   
		We go over some basic notions and definitions which can be found in   \cite{PZ}.   Let $V$ be a set and $E$ be an equivalence relation on $V$. Also, let the set of equivalence classes of $E$ be denoted by $V/E.$  If a set $X \subseteq V$,  is equal to a union of some of the  equivalence classes of $E$ then it is called \textit{E-exact}. Otherwise, $X$ is called \textit{E-inexact} or \textit{E-rough} or simply \textit{rough} when the equivalence relation under consideration is clear from the context. Inexact sets  may be approximated by  two exact sets, the lower and upper  approximations as is respectively defined below:
		\begin{center} 
			$ \low_E (X)	 = \{ x \in V\ | \ [x]_E \subseteq X \} $, 
		\end{center}
		\begin{equation} \label{eq:1}
		\upp_E (X) = \{  x\in V \ | \ [x]_E  \cap X \neq \emptyset \}. 
		\end{equation}
		Equivalently,  we may use a granule based definition instead of a pointwise based definition:
		\begin{center}
			$ \low_E (X)= \bigcup \{ Y\subseteq V/E \ | \ Y \subseteq X  \}, $
		\end{center}
		\begin{equation} \label{eq:2}
		\upp_E (X)= \bigcup \{  Y \subseteq V/E \ | \ Y \cap X \neq \emptyset  \}. 
		\end{equation}

		The pair $(V, E)$ is called an \textit{approximation space}. 	It may be the case that  several equivalence relations are considered over a set. Let $\mathscr{E} $ being a family of equivalence relations over a finite non-empty set $V$.  The pair,  $K = (V, \mathscr{E} )$ is called \textit{knowledge base}.  If $ \mathscr{P} \subseteq \mathscr{E} $, we recall that $\bigcap \mathscr{P}$ is alao an equivalence relation. The intersection of all equivalence relations belonging to $\mathscr{P}$ is denoted by $IND(\mathscr{P}) = \bigcap \mathscr{P}$. This is called  the \textit{indiscernibility relation} over $\mathscr{P}$.\\  
		
		\noindent For two equivalence relations $E_1$ and $E_2,$ we say that $E_1 \leq E_2$ iff $E_1 \subseteq E_2.$ In this case we say that $E_1$ is \emph{finer} than $E_2$ or that $E_2$ is \emph{coarser} than $E_1.$\\
		
		\noindent We recall from \cite{TA} some definitions about different types of roughly definable and undefinable sets. Let $V$ be a set then for $X \subseteq V:$
		
		\vspace{2mm}
		\noindent(i) If $ \low_E (X)\neq \emptyset$ and $\upp_E (X) \neq V,$ then $X$ is called \textit{roughly E-definable.}
		
		\vspace{2mm}
		\noindent (ii) If $ \low_E (X)= \emptyset$ and $\upp_E (X) \neq V,$ then $X$ is called \emph{internally roughly E-undefinable.}
		
		\vspace{2mm}
		\noindent (iii) If $ \low_E (X) \neq \emptyset$ and $\upp_E (X) = V,$ then $X$ is called \emph{externally roughly E-definable.}
		
		\vspace{2mm}
		\noindent (iv) If $ \low_E (X) = \emptyset$ and $\upp_E (X) = V,$ then $X$ is called \emph{totally roughly E-definable.}

		\subsection{Properties Satisfied by Rough Sets} 
		
		In \cite{PZ}, Pawlak enlists the following properties of lower and upper approximations. Let $V$ be a non-empty finite set and $X, Y \subseteq V$. Then, the following holds:\\

		\begin{onehalfspace}
			\noindent $1) \low_E (X) \subseteq X \subseteq \upp_E (X),$
			
			\vspace{2mm}
			\noindent $2) \low_E (\emptyset) = \upp_E (\emptyset) = \emptyset; \quad \low_E (V) = \upp_E (V)= V,$
			
			\vspace{2mm}
			\noindent $ 3) \upp_E (X \cup Y) = \upp_E (X) \cup \upp_E (Y),$
			
			\vspace{2mm}
			\noindent $ 4) \low_E (X \cap Y) = \low_E (X) \cap \low_E (Y),$
			
			\vspace{2mm}
			\noindent $ 5)  X \subseteq Y \Rightarrow \low_E (X) \subseteq \low_E (Y),$
			
			\vspace{2mm}
			\noindent $ 6)  X \subseteq Y \Rightarrow \upp_E (X) \subseteq \upp_E (Y),$
			
			\vspace{2mm}
			\noindent $ 7) \low_E (X\cup Y) \supseteq \low_E (X) \cup \low_E (Y),$
			
			\vspace{2mm}
			\noindent $ 8) \upp_E (X\cap Y) \supseteq \upp_E (X) \cap \upp_E (Y),$
			
			\vspace{2mm}
			\noindent $ 9) \low_E (-X) = -\upp_E (X),$
			
			\vspace{2mm}
			\noindent $ 10) \upp_E (-X) = -\low_E (X),$
			
			\vspace{2mm}
			\noindent $ 11) \low_E (\low_E (X)) = \upp_E (\low_E (X)) = \low_E (X),$
			
			\vspace{2mm}
			\noindent $ 12) \upp_E (\upp_E (X)) = \low_E (\upp_E (X)) = \upp_E (X).$
		\end{onehalfspace}
		
	\subsection{ Dependencies in Knowledge Bases}  
		
		A database can also be represented  in the form of a matrix of \emph{Objects} versus \emph{Attributes} with the entry corresponding to an object attribute pair being assigned the value of that attribute which the object satisfies. From the following definition, we can form equivalence relations on the objects for each given attribute. The set of these equivalence relations can then be used as our knowledge base.
		
		\begin{definition}
			Let $V$ be the set of objects and $P$ be the set of attributes.	Let $ Q \subseteq P$, then V/Q is an equivalence relation on $U$ induced by Q as follows: $x\sim_Qy $ iff $q(x) = q(y)$ for every $q \in Q.$
			
		\end{definition}
		
		\noindent To construct decision rules, we may fix two sets of attributes  called \emph{condition attributes} and  \emph{decision attributes} denoted by ${C}$ and ${D}$ respectively. We then use these to make predictions about the decision attributes based on the condition attributes. \emph{Decision rules} are made by recording which values of decision attributes correlate with which values of condition attributes. As this information can be of considerable size, one of the primary goals of rough set theory is to reduce the number of decision attributes without losing predictive power. A minimal set of attributes which contains the same predictive power as the full set of decision attributes is called a \emph{reduct} with respect to \emph{D}.  
		
		\vspace{2mm}
		\noindent Next we give the definition of the positive region of one equivalence relation with respect to another.

		\begin{definition}
			Let $C$ and $  D $ be equivalence relations on a finite non-empty set $V.$ The \emph{positive region} of the partition $D$ with respect to $C$ is given by, 
			\begin{equation}
			POS_C(D) = \bigcup\limits_{X \in D} \low_C (X),
			\end{equation} 	
			
		\end{definition}

		\begin{definition}
			It is said that  \emph{$ {D} $ depends on ${C} $ in a degree ${k}$}, where 	$0 \leq k \leq 1$, denoted by $C \Rightarrow_{k} D,$ if
			\begin{equation}
			k = \gamma(C,D) = \frac{|POS_C(D)|}{|V|}.
			\end{equation}
			  
		\end{definition} 
		
		\noindent If $k = 1,$ then we say that   $C$ depends totally on $D$ i.e $C \Rightarrow D.$

		\vspace{2mm}
		\noindent Let $K_1 = (V, \mathscr{P})$ and $K_2 = (V, \mathscr{Q}).$  We  now give the definitions dependency of knowledge and then partial dependency. We say that \emph{$\mathscr{Q}$ depends on $\mathscr{P}$} i.e. $ \mathscr{P} \Rightarrow \mathscr{Q}$ iff $IND(\mathscr{P}) \subseteq IND(\mathscr{Q}).$   
		
		\begin{proposition}
			$I_{IND(\mathscr{P})}  \leq I_{IND(\mathscr{Q})}$ iff $ \mathscr{P} \Rightarrow \mathscr{Q}.$	
		\end{proposition}  
		
		\begin{proposition}
			$POS_{IND(\mathscr{P})}IND((\mathscr{Q})) = U$ iff $\mathscr{P} \Rightarrow \mathscr{Q}.$ 
		\end{proposition}

		\noindent Otherwise,  in the above case, $\gamma(IND(\mathscr{P}), IND(\mathscr{Q})) = k <1 $ and then we say that $\mathscr{P} \Rightarrow_k \mathscr{Q}.$\\

	\section{Properties of Successive Approximations} 
	
	Next, we see that in general, approximating with respect to $E_1$ and then approximating the result with respect to $E_2$ gives a different result than if we had done it in the reverse order. That is, successive approximations do not commute. We consider some properties of successive approximations below. 
	
		\begin{proposition}
			Let $V$ be a set and $E_1$ and $E_2$ be equivalence relations on $V.$ Then for $Y \in \mathscr{P}(V),$  the following holds,\\

			\noindent 1. $\textbf{l}_{E_1}(\textbf{l}_{E_2}(Y)) = Z \ \, \, \, \, \not\Rightarrow \  \textbf{l}_{E_2}(\textbf{l}_{E_1}(Y)) = Z, $\\
			2. $\textbf{u}_{E_1}(\textbf{u}_{E_2}(Y)) = Z \ \not\Rightarrow \  \textbf{u}_{E_2}(\textbf{u}_{E_1}(Y)) = Z,$\\
			3. $\textbf{u}_{E_1}(\textbf{l}_{E_2}(Y)) = Z \ \  \not\Rightarrow \  \textbf{l}_{E_2}(\textbf{u}_{E_1}(Y)) = Z,$\\
			4. $\textbf{l}_{E_1}(\textbf{u}_{E_2}(Y)) = Z \ \  \not\Rightarrow \ \textbf{u}_{E_2}(\textbf{l}_{E_1}(Y)) = Z.$\\

		\end{proposition}
		
		\begin{proof}
			We give a counterexample to illustrate the proposition.  Let $V = \{a, b, c, d\}$ and  let $E_1 = \{  \{a, b,c\}, \{  d \}    \}$ and $E_2 = \{  \{a, b\}, \{c, d\}    \}.$
			
			\noindent To illustrate 1., let $Y = \{  a, b, c\}$. Then $\textbf{\textit{l}}_{E_1}(\textbf{\textit{l}}_{E_2}(Y)) = \emptyset$ while $\textbf{\textit{l}}_{E_2}(\textbf{\textit{l}}_{E_1}(Y)) = \{ a, b\}.$
			
			\vspace{2mm}
			\noindent For 2., let $Y = \{a\}$. Then $\textbf{\textit{u}}_{E_1}(\textbf{\textit{u}}_{E_2}(Y)) = \{ a, b, c  \}$ while $\textbf{\textit{u}}_{E_2}(\textbf{\textit{u}}_{E_1}(Y)) = \{a, b, c, d\}.$
			
			\vspace{2mm}
			\noindent For 3., let $Y = \{a, b\}$. Then $\textbf{\textit{u}}_{E_1}(\textbf{\textit{l}}_{E_2}(Y)) = \{ a, b, c  \}$ while $\textbf{\textit{l}}_{E_2}(\textbf{\textit{u}}_{E_1}(Y)) = \{a, b\}.$
			
			\vspace{2mm} 
			\noindent For 4., let $Y = \{a, b, c\}$. Then $\textbf{\textit{l}}_{E_1}(\textbf{\textit{u}}_{E_2}(Y)) = \emptyset $ while $\textbf{\textit{u}}_{E_2}(\textbf{\textit{l}}_{E_1}(Y)) = \{a, b, c, d\}.$
		\end{proof}
	
	\noindent From Properties 1), 5) and 6) of lower and upper approximations in Section 2.1, we immediately get that,\\
	(i) $\textbf{\textit{l}}_{E_1}(\textbf{\textit{l}}_{E_2}(Y)) \subseteq \textbf{\textit{l}}_{E_2}(Y), 
	\textbf{\textit{u}}_{E_1}(\textbf{\textit{u}}_{E_2}(Y)) \supseteq \textbf{\textit{u}}_{E_2}(Y),$\\
	$\textcolor{white}{ggl} \textbf{\textit{u}}_{E_1}(\textbf{\textit{l}}_{E_2}(Y)) \supseteq \textbf{\textit{l}}_{E_2}(Y) \ \text{and}
	\ \textbf{\textit{l}}_{E_1}(\textbf{\textit{u}}_{E_2}(Y)) \subseteq \textbf{\textit{u}}_{E_2}(Y).$\\
	\hfill\\      
	\noindent If we do not know anything more about the relationship between $E_1$ and $E_2$ then nothing further may be deduced. However, if for example  we know that $E_1 \leq E_2$ then the successive approximations are constrained as follows:
	
	\begin{proposition}
		If $E_1 \leq E_2$  then the following properties hold;
		
		\vspace{2mm}
		\noindent (ii) $\textbf{\textit{l}}_{E_1}(\textbf{\textit{l}}_{E_2}(Y)) = \textbf{\textit{l}}_{E_2}(Y) $\\
		(iii) $ \textbf{\textit{l}}_{E_2}(\textbf{\textit{l}}_{E_1}(Y))  \subseteq \textbf{\textit{l}}_{E_2}(Y) $ \\
		(iv) $ \textbf{\textit{u}}_{E_1}(\textbf{\textit{u}}_{E_2}(Y)) \supseteq \textbf{\textit{u}}_{E_1}(Y) $\\
		(v) $ \textbf{\textit{u}}_{E_2}(\textbf{\textit{u}}_{E_1}(Y)) = \textbf{\textit{u}}_{E_2}(Y) $
	\end{proposition}
	
	\begin{proof}
		Straightforward.
	\end{proof}
	
	\begin{proposition}
		Let $V$ be a finite non-empty set and let $E_1$ and $E_2$ be equivalence relations on $V.$ Let $ x\in V.$ Then 	$\textbf{\textit{l}}_{E_1}( \textbf{\textit{u}}_{E_2}(\{x\}))  \subseteq POS_{E_1}({E_2}).$
	\end{proposition}
	
	\begin{corollary}
	Let $V$ be a finite non-empty set and let $E_1$ and $E_2$ be equivalence relations on $V.$	Let $ X \subseteq V.$ Then $POS_{E_1}({E_2}) \cap X \subseteq \bigcup \limits_{x \in X} \textbf{\textit{l}}_{E_1}( \textbf{\textit{u}}_{E_2}(\{x\})).$ 
	\end{corollary}
	
	\begin{corollary}
	Let $V$ be a finite non-empty set and let $E_1$ and $E_2$ be equivalence relations on $V.$	 Then $POS_{E_1}({E_2}) = \bigcup \limits_{x \in V} \textbf{\textit{l}}_{E_1}( \textbf{\textit{u}}_{E_2}(\{x\})).$ 
	\end{corollary}
	
	\begin{center} 
		
		\includegraphics[scale = .5]{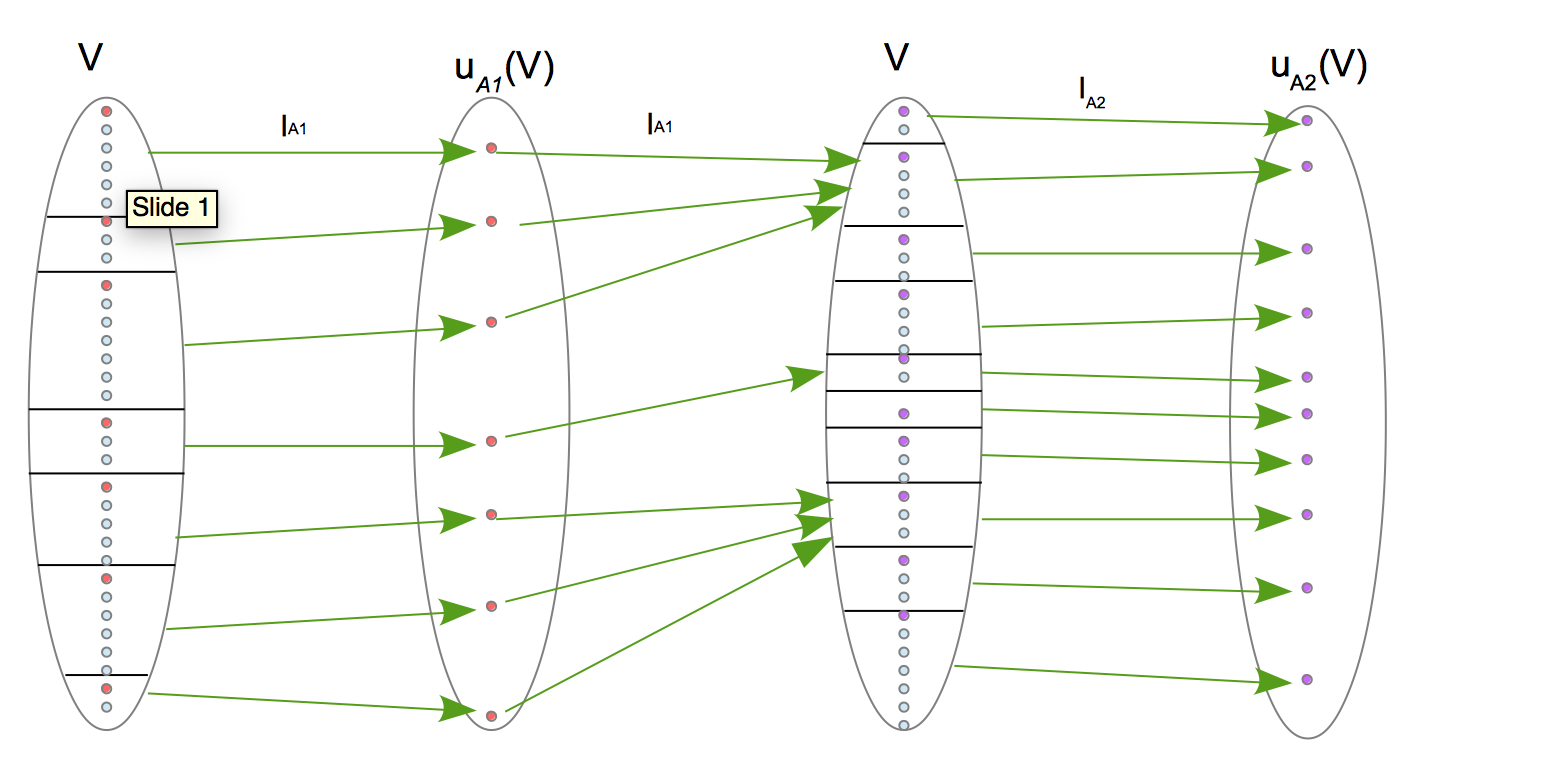}
		\textcolor{white}{aaaaaaa} Figure 4.1:  Illustrates that successive approximations get more coarse when iterated.
		
	\end{center}  
	
	\begin{proposition}
		Let $G$ be a graph with vertex set $V$ and $E$ an equivalence relation on $V.$ Let $S_E$ be the set containing equivalence classes of $E$ and taking the closure under union. Let $ F : \mathscr{P}(V) \rightarrow S_E$ be such that $F(X) = \bigcup\limits_{x \in X}  [x]_E$ and let  $Id: S_E \rightarrow \mathscr{P}(V) $ be such that $Id(Y) = Y$ Then $F$ and $Id$ form a Galois connection.
	\end{proposition}

	\begin{proof}
		
		It is clear from  definitions that both $F$ and $Id$ are  monotone.  We need that for $X \in \mathscr{P}(V)$ and $Y \in S_E, \ F(X) \subseteq Y$ iff $ X \subseteq Id(Y).$  This is also the case because from the definition of $F,$ we have the $X \subseteq F(X).$  
	\end{proof}  
	 
	\noindent \textbf{Remark 3.1}. Successive approximations break the Galois structure of single approximations. We can imagine that single approximations are a kind of sorting on the domain of a structure. We partition objects in the domain into boxes and in each box there is a special member (the lower or upper approximation) which identifies/represents any member in its respective box.  We may say that objects are approximated by  their representative. 
	
	For successive approximations, we have two different sortings of the same domain. Objects are sorted by the first approximation and only their representative members are then sorted by the second approximation. An object is then placed in the box that its representative member is assigned to in the second approximation, even though the object itself may be placed differently if the second approximation alone was used. Hence the errors `add' in some sense.  In Figure 4.1, the final grouping as seen by following successive arrows, may be coarser than both the first and second approximations used singly. An interesting problem is how to correct/minimise these errors. It is also interesting how much of the individual approximations can be reconstructed from knowledge of the combined approximation. In the next section we will investigate this problem.

	\section{Decomposing $L_2L_1$ Approximations} 
	
	What if we knew that a system contained exactly two successive approximations? Would we be able to decompose them into its individual components? Before getting into what we can do and what information can be extracted, we start with an example to illustrate this.\\ 
	
	\noindent \textbf{Notation}: Let $V$ be a finite set. Let a function representing the output of a subset of $V$ when acted on by a lower approximation operator $L_1$ followed by a lower approximation operator $L_2,$ based on the equivalence relations $E_1$ and $E_2$ respectively, be denoted by $L_2L_1$ where $L_2L_1(X) = L_2(L_1(X))$ and $L_2L_1: \mathscr{P}(V) \rightarrow \mathscr{P}(V).$ Similarly, other combinations of successive lower and upper approximations examined will be denoted by $U_2U_1,\ L_2U_1, U_2L_1$  which denotes successive upper approximations, an upper approximation followed by a lower approximation and a lower approximation followed by an upper approximation respectively. 
	
	Sometimes when we know that the approximations are based on equivalence relations $P$ and $Q$ we may use the subscripts to indicate this for example; $L_QL_P.$
	
	Lastly, if for a defined $L_2L_1$ operator there exists a pair of equivalence relation solutions $E_1$ and $E_2$ which are such that the lower approximation operators $L_1$ and $L_2$ are based on them respectively, then we may denote this solution by the pair $(E_1, E_2).$ Also, $(E_1, E_2)$ can be said to produce or generate the operators based on them.\\

	\noindent \textbf{Example 4.1}\\
	Let $V = \{a, b, c, d, e\}.$ Let a function representing the output of a subset of $V$ when acted on by a lower approximation operator $L_1$ followed by a lower approximation operator $L_2,$ which are induced by equivalence relations $E_1$ and $E_2$ respectively and let  $L_2L_1: \mathscr{P}(V) \rightarrow \mathscr{P}(V)$ be as follows:\\
	\hfill\\
	$L_2L_1(\{ \emptyset\}) = \emptyset     \qquad \qquad \qquad \qquad \qquad \qquad \qquad  L_2L_1(\{a, b, c, d, e\}) =  \{a, b, c, d, e\} $ \\
	$L_2L_1(\{a\}) = \emptyset    \qquad \qquad \qquad \qquad \qquad \qquad \qquad  L_2L_1(\{b, c, d, e\})= \{e\}  $\\ 
	$L_2L_1(\{b\}) = \emptyset     \qquad \qquad \qquad \qquad \qquad \qquad \qquad  L_2L_1(\{a, c, d, e\})  = \{c, d, e\} $\\ 
	$L_2L_1(\{c\}) =  \emptyset    \qquad \qquad \qquad \qquad \qquad \qquad \qquad   L_2L_1(\{a, b, d, e\})  = \{e\} $\\  
	$L_2L_1(\{d\}) = \emptyset     \qquad \qquad \qquad \qquad \qquad \qquad \qquad   L_2L_1(\{a, b, c, e\}) = \{a, b \} $\\ 
	$L_2L_1(\{e\}) = \emptyset     \qquad \qquad \qquad \qquad \qquad \qquad \qquad   L_2L_1(\{a, b, c, d\}) = \{a, b\} $\\   
	$L_2L_1(\{a, b\}) = \emptyset     \qquad \qquad \qquad \qquad \qquad \qquad \quad L_2L_1(\{c, d, e\}) = \{e\}  $\\ 
	$L_2L_1(\{a, c\}) =  \emptyset    \qquad \qquad \qquad \qquad \qquad \qquad \quad  L_2L_1(\{b, d, e\}) =  \{e\} $\\ 
	$L_2L_1(\{a, d\}) = \emptyset     \qquad \qquad \qquad \qquad \qquad \qquad \quad  L_2L_1(\{b, c, e\}) = \emptyset $\\ 
	$L_2L_1(\{a, e\}) = \emptyset     \qquad \qquad \qquad \qquad \qquad \qquad \quad  L_2L_1(\{b, c, d\}) = \emptyset  $\\ 
	$L_2L_1(\{b, c\}) = \emptyset     \qquad \qquad \qquad \qquad \qquad \qquad \quad  L_2L_1(\{b, c, d\}) = \emptyset $\\ 
	$L_2L_1(\{b, d\}) = \emptyset     \qquad \qquad \qquad \qquad \qquad \qquad \quad  L_2L_1(\{a, d, e\}) = \{e\}  $\\ 
	$L_2L_1(\{b, e\}) =  \emptyset    \qquad \qquad \qquad \qquad \qquad \qquad \quad  L_2L_1(\{a, c, d\}) = \emptyset $\\ 
	$L_2L_1(\{c, d\}) = \emptyset     \qquad \qquad \qquad \qquad \qquad \qquad \quad  L_2L_1(\{a, b, e\}) = \emptyset $\\ 
	$L_2L_1(\{c, e\}) = \emptyset     \qquad \qquad \qquad \qquad \qquad \qquad \quad  L_2L_1(\{a, b, d\}) = \emptyset $\\ 
	$L_2L_1(\{d, e\}) = \{ e\}     \qquad \qquad \qquad \qquad \qquad \qquad  L_2L_1(\{a, b, c\}) = \{a, b\} $\\ 
	\hfill\\
	We will now try to reconstruct $E_1$ and $E_2.$  The minimal sets in the output are $\{e\}$ and $\{a,b\}.$ Clearly, these are either equivalence classes of $E_2$ or a union of two or more equivalence classes of $E_2.$ Since $\{e\}$ is a singleton it is must be an equivalence class of $E_2.$  So far we have partially reconstructed $E_2$ and it is equal to or finer than $\{ \{a, b\}, \{c, d\}, \{e\}     \}.$ 
	
	Let us consider the pre-images of these sets in  $L_2L_1$ to try to reconstruct $E_1.$ Now, $L_2L_1^{-1}(\{e\}) = \{  \{ d,e \}, \{a, d, e \},  \{ b, d, e\}, \{ c, d, e\}, \{ a, b, d, e\}, \{ b, c,  d, e\} \}.$ We see that this set has a minimum with respect to containment and it is $\{d, e\}.$ Hence  either $\{d,e\}$ is an equivalence class of $E_1$ or both of $\{d\}$ and $\{e\}$ are equivalent classes of $E_1.$ 
	
	Similarly, $L_2L_1^{-1}(\{a, b\}) = \{  \{ a, b, c \}, \{a, b, c, e\}, \{a, b, c, d\} \}.$ We, see that this set has a minimum which is $\{a, b, c\}$ hence either this set is an equivalence class or is a union of equivalence classes in $E_1.$ Now, $L_2L_1^{-1}(\{ c, d, e\}) = \{\{a, c, d, e\}\}.$ Hence, $\{a, c, d, e\}$ also consists of a union of equivalence classes of $E_1.$ Since we know from above that $\{d,e\}$ consists of the union of one or more equivalence classes of $E_1$, this means that $\{a,c\}$ consists of the union of one or more equivalence classes of $E_1$ and $\{b\}$ is an equivalence class of $E_1.$ So far we have that $E_1$ is equal to or finer than $\{ \{a,c\}, \{b\}, \{d, e\}      \} .$
	
	Now we consider if $\{a, c\} \in E_1$ or both of $\{a\}$ and $\{c\}$ are in $E_1.$ We can rule out the latter for suppose it was the case. Then $L_2L_1(\{a, b\})$ would be equal to $\{a, b\}$ since we already have that $\{b\} \in E_1$ and $\{a, b\}$ is the union of equivalence classes in $E_2.$ Since this is not the case we get that $\{a,c\} \in E_1.$ By a similar analysis of $L_2L_1(\{a,c, d\}) \neq \{c, d\}$ but only $\emptyset$ we get that $\{ d, e\} \in E_1.$  Hence, we have fully constructed $E_1$ and $E_1 = \{  \{a, c\}, \{b\}, \{d, e\} \}.$ 
	
	With $E_1$ constructed we can complete the construction of $E_2.$ Recall, that we have that $\{a, b\}$ is a union of equivalence classes in $E_2.$ Suppose that $\{a\} \in E_2.$ Then $L_2L_1(\{a,c\}) $ would be equal to $\{a\}$ since $\{a, c \} \in E_1$ but from the given list we see that it is not. Hence, $\{a, b\} \in E_2.$ Similarly, we recall that $\{c, d \}$ is a union of equivalence classes in $E_2.$ Suppose that $\{d\} \in E_2.$ Then $L_2L_1(\{d, e\})$ would be equal to $\{d, e\}$ since $\{d,e\} \in E_1$ but it is only equal to $\{e\}.$ Hence, $\{c, d\} \in E_2.$ We have now fully reconstructed $E_2$ and $E_2 = \{ \{a, b\}, \{ c, d\}, \{ e\}  \}.$
	\hfill\\
	
	\noindent The next example shows that we cannot always uniquely decompose successive approximations.\\
	\hfill\\ 
	\textbf{Example 4.2}
	
	\noindent  Let $V = \{a, b, c, d\}$ and let $E_1 = \{ \{ a, b\}, \{ c, d\} \}, \ E_2 = \{ \{a, c \}, \{ b, d\} \}$ and $E_3 = \{ \{a, d \}, \{ b, c\} \}.$  We see that $L_1L_2(X) = L_1L_3(X) = \emptyset$ for all $X \in (\mathscr{P}(V) -V)$ and $L_1L_2(X) = L_1L_3(X) = V$ when $X = V.$ Then for all $X \subseteq U,$ $ L_1L_2(X) = L_1L_3(X)$ even though $E_2 \neq E_3.$ Hence, if we are given a double, lower successive approximation on $\mathscr{P}(V)$ which outputs $\emptyset$ for all $X \in (\mathscr{P}(V) -V)$ and $V$ for $X = V$ then we would be unable to say that it was uniquely produced by $L_1L_2$ or $L_1L_3.$ \\
	\hfill\\
	\noindent In the following we start to build to picture of what conditions are needed for the existence of unique solutions for  double, successive approximations. \\
	
	\begin{proposition}
		Let $V$ be a set with equivalence relations $E_1$ and $E_2$  on $V.$ If for each $[x]_{E_1} \in E_1,$ $[x]_{E_1}$ is such that $L_2([x]_{E_1}) = \emptyset$ i.e $[x]_{E_1}$ is either internally $E_2$--undefinable or totally $E_2$--undefinable, then the corresponding approximation operator, $L_2L_1$ on  $\mathscr{P}(V)$ will be such that $L_2L_1([x]_{E_1}) = \emptyset.$ 
	\end{proposition}
	
	\begin{proof}
		Here, $L_1([x]) = \emptyset.$ Hence $L_2L_1([x]) = L_2(\emptyset) = \emptyset.$
	\end{proof}   
	
	\noindent  \textbf {Remark 4.1} We note that the union of $E$-undefinable sets is not necessarily $E$-undefinable. Consider Example 4.2. Here, $\{a, b\}$ and  $\{c, d\}$ are both  totally $E_2$--undefinable but their union, $\{a, b, c, d\}$ is $E_2$--definable.\\

	\noindent \textbf{Algorithm 4.1: For Partial Decomposition of Double Successive Lower Approximations}
	
	\vspace{4mm}
	\noindent Let $V$ be a finite set. Given an input of a  fully defined operator $L_2L_1 : \mathscr{P}(V) \rightarrow \mathscr{P}(V),$  if a solution exists, we can produce a solution $(S, R)$, i.e. where $L_1$ and $L_2$ are the lower approximation operators of equivalence relations $S$ and $R$ respectively,  by performing the following steps:
	
	\vspace{4mm}

	\vspace{3mm}
	\noindent \textbf{1}. Let $J$ be the set of output sets of the given $L_2L_1$ operator. We form the relation $R$ to be such that for $a, b \in V,$ $a \sim_R b \iff (a \in X \iff b\in X)$ for any $X \in J.$ It is clear that $R$ is an equivalence relation.

	\vspace{3mm}
	\noindent \textbf{2}. For each $Y \neq \emptyset$ output set, find the minimum pre-image set with respect to $\subseteq,$ $Y_m,$ such that $L_2L_1(Y_m) = Y$. Collect all these minimum sets in a set $K.$  If there is any non-empty output set $Y,$ such that the minimum $Y_m$ does not exist, then there is no solution to the given operator and we return 0 signifying that no solution exists.
	
	\vspace{3mm}
	\noindent \textbf{3}. Using $K,$ we form the relation $S$ to be such that for $a, b \in V,$ $a \sim_S b \iff (a \in X \iff b\in X)$ for any $X \in K.$ It is clear that $S$ is an equivalence relation.
	
	\vspace{3mm}
	\noindent \textbf{4}. Form the operator $L_RL_S : \mathscr{P}(V) \rightarrow \mathscr{P}(V)$ generated by $(S, R).$ If for all $X \in \mathscr{P}(V)$, the given $L_2L_1$ operator is such that $L_2L_1(X) = L_RL_S(X),$ then $(S, R)$ is a solution proving that a solution exists (note that it is not necessarily unique). Return $(S, R).$ Otherwise, discard $S$ and $R$ and return 0 signifying that no solution exists.\\
	
	\noindent We will  prove the claims in step 2 and step 4 in this section. Next, we prove step 2.

	\begin{proposition} \label{l4}
		Let $V$ be a set and $L_2L_1 : \mathscr{P}(V) \rightarrow \mathscr{P}(V)$ be a given fully defined operator on $\mathscr{P}(V).$ If for $Y \neq \emptyset$ in the range of $L_2L_1,$  there does not exist a minimum set $Y_m,$ with respect to $\subseteq$ such that $L_2L_1(Y_m) = Y,$  then there is no equivalence relation pair solution to the given operator.
		
	\end{proposition}
	
	\begin{proof}
		Suppose to get a contradiction that a solution $(E_1, E_2)$ exists and  there is no minimum set $Y_m$ such that $L_2L_1(Y_m) =Y.$ Since $V$ is finite, then there exists at least two minimal sets $Y_k$ and $Y_l$ say, such that $L_2L_1(Y_s) =Y$ and $L_2L_1(Y_t) =Y.$	Since $Y_s$ and $Y_t$ are minimal sets with the same output after two successive lower approximations, then $Y_s$ and $Y_t$ must each be unions of equivalence classes in $E_1$ which contain $Y.$ Since they are unequal, then WLOG there exists $[a]_{E_1} \in E_1$ which is such that $[a]_{E_1} \in Y_s$ but $[a]_{E_1} \not\in Y_t.$ Since $Y_s$ is minimal, then $[a]_{E_1} \cap Y \neq \emptyset$ (or else $L_2L_1(Y_s) = L_2L_1(Y_s - [a]_{E_1}) = Y$). So let $ x \in [a]_{E_1} \cap Y.$ Then $Y_t \not\supseteq x$ which contradicts $Y_t \supseteq Y.$
	\end{proof} 
	
	\noindent We now prove three lemmas on the way to proving the claim in step 4. 
	
	\begin{lemma} \label{p15}
		Let $V$ be a set and $L_2L_1 : \mathscr{P}(V) \rightarrow \mathscr{P}(V)$ be a given fully defined operator on $\mathscr{P}(V).$   Let $R$ and $S$ be equivalence relations defined on $V$ as constructed in the previous algorithm.  If $(E_1, E_2)$ is a solution of  $L_2L_1$ then $E_2\leq R$ and $E_1 \leq S.$
		
	\end{lemma}
	
	\begin{proof}
		We first prove 	$E_2\leq R.$ Now the output set of a non-empty set in $\mathscr{P}(V)$ is obtained by first applying the lower approximation $L_1$ to it and and after applying the lower approximation, $L_2$ to it. Hence by definition of $L_2,$ the non-empty output sets are unions of equivalence classes of the equivalence relation which corresponds to $L_2.$   If $a$ is in an output set but $b$ is not then they cannot belong to the same equivalence class of  $E_2$ i.e. $a \not\sim_R b$ implies that $a \not\sim_{E_2} b.$ Hence $E_2\leq R. $ 
		
		Similarly, the minimal pre-image, $X$ say,  of a non-empty output set which is a  union of equivalence classes in $E_2,$ has to be a union of equivalence classes in $E_1.$ For suppose it was not. Let $Y = \{y \in X\ | \ [y]_{E_1} \not\subseteq X\}.$ By assumption, $Y \neq \emptyset.$ Then $L_1 (X) = L_1 (X - Y).$ Hence $L_2L_1(X) = L_2L_1 (X - Y)$ but $|X - Y| < |X|$ contradicting minimality of $X$. Therefore, if $a$ belongs to the minimal pre-image of a non-empty output set but $b$ does not belong to it, then $a$ and $b$ cannot belong to the same equivalence class in $E_1$ i.e. $a \not\sim_S b$ which implies that $a \not\sim_{E_1} b.$ Hence $E_1\leq S.$  
	\end{proof}

	\noindent \textbf{Remark 4.2} The above lemma implies that for a given $L_2L_1$ operator on $\mathscr{P}(V)$ for a set $V,$ that the pair of solutions given by the algorithm $S$ and $R$   for corresponding to $L_1$ and $L_2$ are the coarsest solutions for $E_1$ and $E_2$ which are compatible with the given, fully defined $L_2L_1$ operator. That is, for any other possible solutions, $E_1$ and $E_2$ to the given $L_2L_1$ operator, $E_1 \leq S$ and $E_2 \leq R.$\\

	\begin{lemma} \label{l2}
		Let $V$ be a finite set and  $L_2L_1 : \mathscr{P}(V) \rightarrow \mathscr{P}(V)$ be a fully defined operator. If there exists  equivalence pair solutions to the operator $(E_1, E_2)$ which is such that there exists  $[x]_{E_2}, [y]_{E_2} \in E_2,$ such that $[x]_{E_2} \neq [y]_{E_2}$ and $\textbf{u}_{E_1}([x]_{E_2}) = \textbf{u}_{E_1}([y]_{E_2}), $ then there exists another solution, $(E_1, H_2)$, where $H_2$ is an equivalence relation formed from $E_2$ by combining $[x]_{E_2}$ and $[y]_{E_2}$ and all other elements are as in $E_2.$ That is, $[x]_{E_2} \cup [y]_{E_2} = [z] \in H_2$ and if $[w] \in E_2$ such that $[w] \neq [x]_{E_2}$ and $[w]_{E_2} \neq [y]_{E_2},$ then $[w] \in H_2.$
		
	\end{lemma}
	
	\begin{proof}
		Suppose that $(E_1, E_2)$ is a solution of a given $L_2L_1$ operator and $H_2$ is as defined above. Now, $L_2L_1(X) = Y $ iff the union of $E_1$-equivalence classes in $X$ contains the union of $E_2$-equivalence classes which is equal to $Y.$ So, in the $(E_1, H_2)$ solution, the only way that $L_{H_2}L_{E_1}(X)$ could be different from $L_{E_2}L_{E_1}(X)$(which is $=L_2L_1(X)$) is if (i) $[x]_{E_2}$ is contained in $L_{E_2}L_{E_1}(X)$ while  $[y]_{E_2}$ is not contained in $L_{E_2}L_{E_1}(X)$ or if (ii) $[y]_{E_2}$ is contained in $L_{E_2}L_{E_1}(X)$ while  $[x]_{E_2}$ is not contained in $L_{E_2}L_{E_1}(X).$ This is because in $H_2,$ $[x]_{E_2}$ and $[y]_{E_2}$ always occur together in an output set if they are in it at all (recall that output sets are unions of equivalence classes) in the equivalence class of $[z]= [x]_{E_2} \cup [y]_{E_2} $ and all the other equivalence classes of $H_2$ are the same as in $E_2.$ However, neither (i) nor (ii) is the case since $\textit{\textbf{u}}_{E_1}([x]_{E_2}) = \textit{\textbf{u}}_{E_1}([y]_{E_2}).$ That is, the equivalence classes of $[x]_{E_2}$ are contained by exactly the same union of equivalences in $E_1$ which contains $[y]_{E_2}.$ Thus, any set $X$ which contains a union of $E_1$-equivalences which contains $[x]_{E_2}$ also must contain $[y]_{E_2}$ and therefore $[z]_H.$ Hence, if $(E_1, E_2)$ is a solution for the given vector, then so is $(E_1, H_2).$
	\end{proof}
	
	\begin{lemma} \label{l3}
		Let $V$ be a finite set and  $L_2L_1 : \mathscr{P}(V) \rightarrow \mathscr{P}(V)$ be a fully defined operator. If there exists  equivalence pair solutions to the operator $(E_1, E_2)$ which is such that there exists  $[x]_{E_1}, [y]_{E_1} \in E_2,$ such that $[x]_{E_1} \neq [y]_{E_1}$ and $\textbf{u}_{E_2}([x]_{E_1}) = \textbf{u}_{E_2}([y]_{E_1}), $ then there exists another solution, $(H_1, E_2)$, where $H_1$ is an equivalence relation formed from $E_1$ by combining $[x]_{E_2}$ and $[y]_{E_2}$ and all other elements are as in $E_1.$ That is, $[x]_{E_1} \cup [y]_{E_1} = [z] \in H_1$ and if $[w] \in E_2$ such that $[w] \neq [x]_{E_1}$ and $[w]_{E_1} \neq [y]_{E_1},$ then $[w] \in H_1.$ 
	\end{lemma}

	\begin{proof}
		
		Suppose that $(E_1, E_2)$ is a solution of a given $L_2L_1$ operator and $H_1$ is as defined above. Now, $L_2L_1(X) = Y $ iff the union of $E_1$-equivalence classes in $X$ contains the union of $E_2$-equivalence classes which is equal to $Y.$ So, in the $(H_1, E_2)$ solution, the only way that $L_{E_2}L_{H_1}(X)$ could be different from $L_{E_2}L_{E_1}(X)$(which is $=L_2L_1(X)$) is if the union of equivalence classes in $X$ which is needed to contain $Y,$ (i) contains $[x]_{E_2} $ but not $[y]_{E_2}$ or (ii) contains $[y]_{E_2} $ but not $[z]_{E_2}.$ However, this is not the case since  $\textit{\textbf{u}}_{E_2}([x]_{E_1}) = \textit{\textbf{u}}_{E_2}([y]_{E_1}).$ That is, $[x]_{E_1}$ intersects exactly the same equivalence classes in $E_2$ as $[y]_{E_1}.$ So if $[x]_{E_1}$ is needed to contain an equivalence class in $E_2,$ then $[y]_{E_1}$ is also needed. In other words, if $L_2L_1(X) = Y,$ then for any minimal set such $Y_m \subseteq X$ such that $L_2L_1(Y_m) = Y,$ $[x]_{E_1}$ is contained in $Y_m$ iff $[y]_{E_1}$ is contained in $Y_m$ iff $[z] \in H_1$ is contained in $Y_m.$  Hence, if $(E_1, E_2)$ is a solution for the given vector, then so is $(H_1, E_2).$
	\end{proof}
	
	\noindent We now have enough to be able to prove the claim in step 4 of Algorithm 4.1 (actually we prove something stronger because we also show conditions which the solutions of the algorithm must satisfy).

	\begin{theorem} \label{t2}
		Let $V$ be a finite set and  $L_2L_1 : \mathscr{P}(V) \rightarrow \mathscr{P}(V)$ be a fully defined operator. If there exists  equivalence pair solutions to the operator, then there exists solutions $(E_1, E_2)$ which satisfy,  
		
		\vspace{2mm} 
		\noindent (i) for each $[x]_{E_2}, [y]_{E_2} \in E_2,$ if $[x]_{E_2} \neq [y]_{E_2}$ then $\textbf{u}_{E_1}([x]_{E_2}) \neq \textbf{u}_{E_1}([y]_{E_2}), $  
		
		\vspace{2mm}
		\noindent (ii) for each $[x]_{E_1}, [y]_{E_1} \in E_1,$ if $[x]_{E_1} \neq [y]_{E_1}$ then $\textbf{u}_{E_2}([x]_{E_1}) \neq \textbf{u}_{E_2}([y]_{E_1})$.
		
		\vspace{2mm}
		\noindent Furthermore, $E_1 = S$ and $E_2 = R$ where $(S, R)$ are the solutions obtained by applying Algorithm 4.1 to the given $L_2L_1$ operator.
	\end{theorem} 	 
	
	\begin{proof}
		
		Suppose that there exists a solution $(C, D).$ Then, either $(C, D)$ already satisfies condition (i) and condition (ii) or it does not. If it does,  take $ (E_1, E_2) = (C, D).$ If it does not satisfy condition (i) then use repeated applications of Lemma \ref{l2} until we arrive at an $(C, E_2)$ solution which does.  Similarly, if $(C, E_2)$ does not satisfy condition (ii), use repeated applications of Lemma \ref{l3} until it does. Since $\mathscr{P}(V)$ is finite this will take at most finite applications of the lemmas until  we obtain a solution, $(E_1, E_2)$ which satisfies the conditions of the theorem.   Since there is a solution, using Proposition \ref{l4} we will at least be able to reach step 4 of Algorithm 4.1. So let $S$ and $R$ be the relations formed by the algorithm after step 3.  Next, we will show that $E_1= S$ and $E_2= R.$ Now, by Lemma \ref{p15}, we have that $E_1 \leq S$ and $E_2 \leq R.$

		Consider the output sets  of the given $L_2L_1$ operator. It is clear that these sets are unions of one or more equivalence classes of $E_2.$ Let $[y]_{E_2} \in E_2$ then $L_2L_1(\textbf{\textit{u}}_{E_1}([y]_{E_2})) \supseteq [y]_{E_2}.$ 
		
		\vspace{2mm}
		\noindent \textbf{Claim 1:} $L_2L_1(\textbf{\textit{u}}_{E_1}([y]_{E_2})) $ is the minimum output set of $L_2L_1$ such that it contains $[y]_{E_2}$ and $\textbf{\textit{u}}_{E_1}([y]_{E_2})$ is the minimum set $X$ such that $L_2L_1(X) \supseteq [y]_{E_2}.$  
		
		\vspace{2mm}
		To see this we first note that $L_2L_1$ is a monotone function on $\mathscr{P}(V)$ since $L_1$ and $L_2$ are monotone operators and $L_2L_1$ is the composition of them. Then, if we can show that $\textbf{\textit{u}}_{E_1}([y]_{E_2})$  is the minimum set $X \in \mathscr{P}(V),$ such that   $L_2L_1(X) \supseteq [y]_{E_2},$ then    $L_2L_1(\textbf{\textit{u}}_{E_1}([y]_{E_2}))$ will be the minimum set output set which contains $[y]_{E_2}.$   This is true because for   $L_2L_1(X) \supseteq [y]_{E_2},$ then $L_1(X)$ must contain each member of $[y]_{E_2}.$  We note that the range of $L_1$ contains only unions of equivalence classes of $E_1$ (counting the emptyset as a union of zero sets). Hence for $L_1(X)$ to contain each element of $[y]_{E_2},$ it must contain each equivalence class in $E_1$ which contains any of these elements. In other words, it must contain $\textbf{\textit{u}}_{E_1}([y]_{E_2}).$ Suppose that $X$ is such that $X \not\supseteq \textbf{\textit{u}}_{E_1}([y]_{E_2})$ and $ L_2L_1(X) \supseteq [y]_{E_2}.$ Then for some $v \in [y]_{E_2},$ $v$ is not in $X$ and so  $\textbf{\textit{u}}_{E_1}([v]_{E_2})\not\in L_1(X).$  Hence $L_2L_1(X) \not\supseteq v $ and so does not contain $[y]_{E_2}$ which is a contradiction. 
		
		\vspace{2mm}
		\noindent \textbf{Claim 2:} $L_2L_1(\textbf{\textit{u}}_{E_1}([y]_{E_2})) $ is not the minimum output set with respect to containing  any other $[z]_{E_2} \neq [y]_{E_2}.$ 
		
		\vspace{2mm}
		Suppose that for some $[z]_{E_2} \neq [y]_{E_2} \in E_2,$ that $L_2L_1(\textbf{\textit{u}}_{E_1}([y]_{E_2})) $  is the minimum  output set containing $[z]_{E_2}.$ Then by the previous Claim, we get that $L_2L_1(\textbf{\textit{u}}_{E_1}([y]_{E_2}))  = L_2L_1(\textbf{\textit{u}}_{E_1}([z]_{E_2}))$ and  that $\textbf{\textit{u}}_{E_1}([y]_{E_2}) \supseteq \textbf{\textit{u}}_{E_1}([z]_{E_2})$. But since $\textbf{\textit{u}}_{E_1}([y]_{E_2})$ is the minimum set such that $L_2L_1(X) \supseteq [y]_{E_2}$, then the stated equality also gives us that $\textbf{\textit{u}}_{E_1}([y]_{E_2}) \subseteq \textbf{\textit{u}}_{E_1}([z]_{E_2}).$ Hence we have $\textbf{\textit{u}}_{E_1}([y]_{E_2}) = \textbf{\textit{u}}_{E_1}([z]_{E_2})$ which is a contradiction to the assumption of condition (i) of the theorem.
		
		Now we can reconstruct $E_2$ by relating elements which always occur together in the output sets. That is,  $a \sim_R b \iff \   (a \in X \iff b \in X)$ for each $X$ in the range of $L_2L_1$. From the previous proposition we have that $E_2 \leq R.$ We claim that $R \leq E_2,$ hence $R = E_2.$ To show this, suppose that it is not the case. Then there exists $a, b \in V$ such that $ a \sim_R b$ but $a \not \sim_{E_2} b.$ By Claim 1, $L_2L_1(\textbf{\textit{u}}_{E_1}([a]_{E_2}))$ is the minimum set which contains $[a]_{E_2}$ and since $a \sim_R b$ then it must contain $b$, and consequently $[b]_{E_2}$ as well. Similarly by Claim 1, $L_2L_1(\textbf{\textit{u}}_{E_1}([b]_{E_2})) $ is the minimum set which contains $[b]_{E_2}$ and since $a \sim_R b$ then it must contain $a,$ and consequently $[a]_{E_2}$ as well. By minimality we therefore have both $ L_2L_1(\textbf{\textit{u}}_{E_1}([a]_{E_2})) \subseteq L_2L_1(\textbf{\textit{u}}_{E_1}([b]_{E_2}))$ and $L_2L_1(\textbf{\textit{u}}_{E_1}([a]_{E_2})) \supseteq L_2L_1(\textbf{\textit{u}}_{E_1}([b]_{E_2}))$ which implies that $L_2L_1(\textbf{\textit{u}}_{E_1}([a]_{E_2})) = L_2L_1(\textbf{\textit{u}}_{E_1}([b]_{E_2})).$ This contradicts Claim 2 since $[a]_{E_2} \neq [b]_{E_2} \in E_2.$ Hence, $E = R$ and we can reconstruct $E_2$ by forming the equivalence relation $R$ which was defined by using the output sets.
		
		It remains to reconstruct $E_1.$ Next, we list the pre-images of the minimal output sets which contain $[y]_{E_2}$ for each $[y]_{E_2}$ in $E_2$ and by Claim 1 this exists and is equal to  $\textbf{\textit{u}}_{E_1}([y]_{E_2}).$ This implies that each such set  is the union of some of the equivalence classes of $E_1.$ Now using this pre-image list we relate elements of $V$ in the following way: $a \sim_S b \iff \ (a\in X \iff b \in X)$ for each $X$ in the pre-image list. From the previous proposition we have that $E_1 \leq S.$  We claim that $S \leq E_1$ and hence $S = E_1.$ Suppose that it was not the case. That is, there exists $a, b \in V$ such that $a \sim_S b$ but $a \not\sim_{E_1} b.$ Hence $[a]_{E_1} \neq [b]_{E_1}.$ By condition (ii) of the theorem, we know that   $\textit{\textbf{u}}_{E_2}([a]_{E_1}) \neq \textit{\textbf{u}}_{E_2}([b]_{E_1}).$ So WLOG suppose that $d \in \textit{\textbf{u}}_{E_2}([a]_{E_1})$ but $d \not\in \textit{\textbf{u}}_{E_2}([b]_{E_1}).$ Since these sets are unions of equivalence classes in $E_2$ this implies that 1), $[d]_{E_2} \subseteq \textit{\textbf{u}}_{E_2}([a]_{E_1})$ and  2) $[d]_{E_2} \cap \textit{\textbf{u}}_{E_2}([b]_{E_1}) = \emptyset.$ Now by Claim 1, $\textit{\textbf{u}}_{E_1}([d]_{E_2})$ is the minimum set, $X$ such that  $L_2L_1(X)$  contains $[d]_{E_2}$ and so is on the output list from which the Relation $S$ was formed. However, 1) implies that this set contains $a$ while 2) implies that this set does not contain $b.$ This contradicts $a \sim_S b.$ Hence $S= E_1$ and we can construct $E_1$ by constructing $S.$ The result is shown.
	\end{proof}

	\noindent Next we give, a graph-theoretic equivalence of the theorem but we first define a graph showing the relationship between two equivalence relations on a set. 
	
	\begin{definition}
		Let $C$ and $D$ be two equivalence relations on a set $V.$ 	Form a bipartite graph $B(C, D) = (G,E),$ where the nodes $G$ is such that $G = \{ [u]_C \ | \  [u]_C \in C  \} \cup \{ [u]_D \ | \ [u]_D \in D \}$ and the edges $E$ are such that $E = \{  ([u]_C, [v]_D)\ | \ \exists \ x \in V : \ x \in [u]_C \ \text{and} \ x \in [v]_D      \}.$ We call this the \textbf{incidence graph} of the pair $(C, D).$	
	\end{definition} 
	
	\begin{theorem}
		Let $V$ be a finite set and let $L_2L_1: \mathscr{P}(V) \rightarrow \mathscr{P}(V)$ be a given fully defined operator on $\mathscr{P}(V).$ If there exists solutions $(E_1, E_2)$ then the incidence graph of $E_1$ and $E_2,$ $B(E_1, E_2), $ is such that  there are no compete bipartite subgraphs as components other than edges (or $K_2$).
	\end{theorem}
	
	\begin{proof}
		This is a direct translation of the previous theorem graph-theoretically. Suppose that the incidence graph of $E_1$ and $E_2$, $B(E_1, E_2),$ contains a complete bipartite subgraph as a component. Then the partition corresponding to $E_2$ violates Condition (i) of the theorem and the partition corresponding to $E_1$ violates condition (ii) of the theorem.  	 
	\end{proof}  
	
	\begin{corollary}
		Let $V$ be a finite set and $L_2L_1 : \mathscr{P}(V) \rightarrow \mathscr{P}(V)$ be a given defined operator. If $(E_1, E_2)$ is a unique solution for the operator then $|E_1| < 2^{|E_2|} $ and $|E_2| < 2^{|E_1|}. $  
	\end{corollary}	
	
	\begin{proof}
		This follows directly from the conditions since in the incidence graph of a unique solution $(E_1, E_2),$  each equivalence class in $E_1$ is mapped to a unique non-empty subset of equivalence classes in $E_2$ and vice versa.
	\end{proof}

	\noindent The next natural question is, without assuming conditions on the equivalence relations, are there instances when the algorithm produces a unique solution?  Example 4.1 is an example of a unique decomposition of a given $L_2L_1$ operator. So this leads naturally to the next question. What conditions result in a unique solution to a given $L_2L_1?$ Can we find  characterising features of  the  pairs of equivalences relations which give a unique $L_2L_1$ operator?    
	
	We note that the algorithm always produces a solution for a fully defined $L_2L_1$ operator which has at least one solution. Hence, if there is a unique solution then these pairs of equivalence relations satisfy the conditions of Theorem \ref{t2}. Recall that in Example 4.2, we were given an $L_2L_1$ operator defined on $\mathscr{P}(V)$ for $ V = \{a, b, c, d\}$ such that $L_2L_1(X) = \emptyset$ for all $X \neq V$ and $L_2L_1(V) = V.$ This example shows us that in addition to a solution which would satisfy the conditions of the theorem, which applying the algorithm gives us;  $E_1 = \{\{a, b, c, d\}  \} $ and $E_2 = \{\{a, b, c, d\}  \}\} ,$ we also have solutions of the form $E_1 = \{ \{a, b\}, \{c, d\}  \} $ and $E_2 = \{ \{a, c\}, \{b, d \}\}$ or $E_1 = \{ \{a, b\}, \{c, d\}  \} $ and $E_2 = \{ \{a, d\}, \{b, c \}\}$ amongst others. In Lemma \ref{p15}, we showed that the solution given by the algorithm is the coarsest pair compatible with a given defined $L_2L_1$ operator. We now try to find a condition such that after applying the algorithm, we may deduce whether or not the $(S, R)$ solution is unique.   This leads us to the next section. 
	 
	\subsection {Characterising Unique Solutions}
	
	\begin{theorem} \label{t3} 
		Let $V$ be a finite set and let $L_2L_1: \mathscr{P}(V) \rightarrow \mathscr{P}(V)$ be a fully defined operator on $\mathscr{P}(V).$ If $(S, R)$ is returned by Algorithm 4.1,  then $(S, R)$ is the unique solution of the operator iff the following holds:
		
		\vspace{2mm}
		\noindent (i) For any $[x]_R \in R,$ 
		there exists $[z]_S \in S$ such that, $ |[x]_R \cap [z]_S| = 1.$ \\
		(ii) For any $[x]_S \in S,$ 
		there exists $[z]_R \in R$ such that, $ |[x]_S \cap [z]_R| = 1.$ 
		
		\vspace{2mm}
	\end{theorem}	 
	
	\begin{proof}
		We prove $\Leftarrow$  direction first. So assume the conditions. We note that by Lemma \ref{p15}, any other solutions, $(E_1,  E_2)$ to the given $L_2L_1$ operator must be coarser than $(S, R).$ Thus, if there is another solution to the given $L_2L_1$ operator, $(E_1, E_2)$ then at least one of $E_1 < S,$ $E_2 < R$ must hold.
		
		First we assume to get a contradiction that there exists a solution $(E_1, E_2)$ which is such  that $E_1 < S.$ That is,  $E_1$ contains a splitting of at least one of the equivalences classes of $S,$ say $[a]_S.$ Hence $|[a]_S| \geq 2.$ By assumption there exists a $[z]_R \in R$  such that $|[a]_S \cap [z]_R| = 1.$ Hence there is a $[z]_{E_2} \in E_2$ such that $|[a]_S \cap [z]_{E_2}| = 1$ since $E_2 \leq R.$ Call the element in this intersection $v$ say. We note that $[v]_{E_2} = [z]_{E_2}.$  Now as $[a]_S$ is spilt into smaller classes in $E_1,$ $v$ must be in one of these classes, $[v]_{E_1}.$   Consider the minimal pre-image of the minimal output set of $L_2L_1$ which contains $[v]_R.$ Call this set $Y_{(S, R)}.$ For the solution $(S, R),$ $Y_{(S, R)}$ contains all of $[a]_S$ since $v \in [a]_S.$ But for the solution  $(E_1, E_2),$ the minimal pre-image of the minimal output set of $L_2L_1$ which contains $[v]_R,$ $Y_{(E_1, E_2)}, $ is such that $Y_{(E_1, E_2)} = (Y_S - [a]_s) \cup [v]_{E_1} \neq Y_S.$  Hence the output list for $(E_1, S)$ is different from the given one which is a contradiction.               
		
		Next, suppose to get a contradiction there exists a solution $(E_1, E_2)$ which is such  that $E_2 < R.$ That is, $E_2$ contains a splitting of at least one of the equivalences classes of $R,$ say $[a]_R.$ Hence $|[a]_R| \geq 2.$ By assumption there exists a $[z]_S \in S$  such that $|[a]_R \cap [z]_S| = 1.$ Hence there is a $[z]_{E_1} \in E_1$ such that $|[a]_R \cap [z]_{E_1}| = 1$ since $E_1 \leq S.$ Call the element in this intersection $v$ say. We note that $[v]_{E_1} = [z]_{E_1}.$ Now as $[a]_R$ is spilt into smaller classes in $E_2,$ $v$ must be in one of these classes, $[v]_{E_2}.$ Consider the set $[a]_R - [v]_{E_2}.$  The minimal pre-image of the minimal output set which contains this set in the $(S,R)$ solution,  $Y_{(S, R)}$  contains $[v]_S$ since here the minimal output set which contains  $([a]_R - [v]_{E_2}),$ must contain all of $[a]_R$ which contains $v.$ If $(E_1, E_2)$ were the solution then the minimal pre-image of the minimal output set which contains  $([a]_R - [v]_{E_2}),$ $Y_{(E_1, E_2)},$ would not contain $[v_s]$ since $([a]_R - [v]_{E_2}) \cap [v]_S = \emptyset.$ That is, $Y_{(E_1, E_2)}   \neq Y_S.$  Hence the output list for $(E_1, E_2)$ is different from the given one which is a contradiction.

		Now we prove $\Rightarrow $ direction. Suppose that $(E_1, E_2)$ is the unique solution, and assume that the condition does not hold.  By Theorem \ref{t2},  $(E_1, E_2) = (S, R).$ Then either there exists an $[x]_R \in R$ such that for all $[y]_S \in S$ such that $[x]_R \cap [y]_S \neq \emptyset$  we have that $|[x]_R \cap [y] _S| \geq 2$ or  there exists an $[x]_S \in S$ such that for all $[y]_R \in R$ such that $[x]_S \cap [y]_R \neq \emptyset $ we have that $|[x]_S \cap [y] _R| \geq 2.$   
		
		We consider the first case. Suppose that $[x]_R$ has non-empty intersection with with $n$ sets in $S.$ We note that $n \geq 1.$ Form a sequence of these sets; $S_1,...S_n. $ Since $|[x]_R \cap S_i| \geq 2$ for each $i$ such that $i= 1,...n,$ let $\{a_{i1}, a_{i2} \}$ be in $[x]_R \cap S_i$ for each $i$ such that $i= 1,...n.$ We split $[x]_R$ to form a finer $E_2$ as follows: Let $P = \{a_{i1}\ | \ i= 1,...n\}$  and $Q = [x]_R - P$ be equivalence classes in $E_2$  and for the remaining equivalence classes in $E_2,$ let $[y] \in E_2$ iff $[y] \in R$ and  $[y]_R \neq [x]_R.$  Now, $L_RL_S(X) = Y $ iff the union of $S$-equivalence classes in $X$ contains the union of $R$-equivalence classes which is equal to $Y.$ So, for the $(S, E_2)$ solution, the only way that $L_{E_2}L_S(X)$ could be different from $L_RL_S(X)$  is if there is a union of $S$-equivalence classes in $X$ which contain $P$ but not $Q$ or which contain $Q$ but not $P$ (since $P$ and $Q$ always occur together as $[x]_R$ for the $(S, R)$ solution). However, this is not the case as follows. Since $P$ and $Q$ exactly spilt all of the equivalence classes of $S$ which have non-empty intersection with $[x]_R,$ we have that $\textit{\textbf{u}}_S(P) = \textit{\textbf{u}}_S(Q).$ That is, $P$ intersects exactly the same equivalence classes of $S$ as $Q.$ Therefore,  $P$ is contained by exactly the same union of equivalence classes in $S$ as $Q.$ Therefore, a union of $S$-equivalence classes in $X$ contains $P$ iff it contains $Q$ iff its contains $[x]_R.$ Hence, $L_RL_S(X) = L_{E_2}L_S(X)$ for all $X \in \mathscr{P}(V)$ and  if $(S, R)$ is a solution for the given vector, then so is $(S, E_2)$ which is a contradiction of assumed uniqueness of $(S, R).$

		We consider the second case. Suppose that $[x]_S$ has non-empty intersection with with $n$ sets in $R.$ We note that $n \geq 1.$ Form a sequence of these sets; $R_1, \dots R_n. $ Since $|[x]_S \cap R_i| \geq 2$ for each $i$ such that $i= 1, \dots n,$ let $\{a_{i1}, a_{i2} \}$ be in $[x]_S \cap R_i$ for each $i$ such that $i= 1, \dots n.$ We split $[x]_S$ to form a finer $E_1$ as follows: Let $P = \{a_{i1}\ | \ i= 1, \dots n\}$ be one equivalence class and let $Q = [x]_R - P$ be another and for any $[y]_S \in S$ such that $[y]_S \neq [x]_S, $ let $[y] \in E_1$ iff $[y] \in S.$ Again, $L_RL_S(X) = Y $ iff the union of $S$-equivalence classes in $X$ contains the union of $R$-equivalence classes which is equal to $Y.$ So, for the $(E_1, R)$ solution, the only way that $L_RL_{E_1}(X)$ could be different from $L_RL_S(X)$  is if (i) $P$ is contained in $L_RL_S(X)$ while  $Q$ is not contained in $L_RL_S(X)$ or  (ii) $Q$ is contained in $L_RL_S(X)$ while  $P$ is not contained in $L_RL_S(X).$ Since $P$ and $Q$ spilt all of the equivalence classes of $R$ which have non-empty intersection with $[x]_S,$ this implies that $\textit{\textbf{u}}_R(P) = \textit{\textbf{u}}_R(Q).$ That is, $P$ and $Q$ intersect exactly the same equivalence classes of $R.$ So if $P$ is needed to contain an equivalence class in $R$ for the $(S, R)$ solution, then $Q$ is also needed. In other words, if $L_2L_1(X) = Y,$ then for any minimal set such $Y_m \subseteq X$ such that $L_2L_1(Y_m) = Y,$ $P$ is contained in $Y_m$ iff $Q$ is contained in $Y_m$ iff $[x]_S$ is contained in $Y_m.$ Hence, $L_RL_S(X) = L_RL_{E_1}(X)$ for all $X \in \mathscr{P}(V)$ and  if $(S, R)$ is a solution for the given vector, then so is $(E_1, R)$ which is a contradiction of assumed uniqueness of $(S, R).$
	\end{proof} 
	
	\noindent The following theorem sums up the results of Theorem \ref{t2} and Theorem \ref{t3}. 
	
	\vspace{2mm}
	\begin{theorem} \label{t4}
		Let $V$ be a finite set and let $L_2L_1: \mathscr{P}(V) \rightarrow \mathscr{P}(V)$ be a fully defined successive approximation operator on $\mathscr{P}(V).$   If $(E_1, E_2)$ is a solution of the operator then it is the unique solution iff the following holds:

		\vspace{2mm}
		\noindent (i) For each $[x]_{E_2}, [y]_{E_2} \in E_2,$ if $[x]_{E_2} \neq [y]_{E_2}$ then $\textbf{u}_{E_1}([x]_{E_2}) \neq \textbf{u}_{E_1}([y]_{E_2}), $  
		
		\vspace{2mm}
		\noindent (ii) For each $[x]_{E_1}, [y]_{E_1} \in E_1,$ if $[x]_{E_1} \neq [y]_{E_1}$ then $\textbf{u}_{E_2}([x]_{E_1}) \neq \textbf{u}_{E_2}([y]_{E_1})$.
		
		\vspace{2mm}
		\noindent (iii) For any $[x]_{E_2} \in E_2,$ 
		there exists $[z]_{E_1} \in E_1$ such that, $ |[x]_{E_2} \cap [z]_{E_1}| = 1.$ 
		
		\vspace{2mm}
		\noindent(iv) For any $[x]_{E_1} \in E_1,$ 
		there exists $[z]_{E_2} \in E_2$ such that, $ |[x]_{E_1} \cap [z]_{E_2}| = 1.$

	\end{theorem}

	\noindent \textbf{Remark 4.3:} If an equivalence relation pair satisfies the conditions of  Theorem \ref{t2}, then the $L_2L_1$ operator based on those relations would be such that if there exists other solutions then they would be finer pairs of equivalence relations.   On the other hand, if an equivalence relation pair satisfies the conditions of  Theorem \ref{t3}, then the $L_2L_1$ operator based on those relations would be such that if there exists other solutions then they would be coarser pairs of equivalence relations.  Hence, if an equivalence relation pair satisfies the conditions of both Theorem \ref{t2} and Theorem \ref{t3}, then the $L_2L_1$ operator produced by it is unique.
	
	\begin{corollary} \label{c2}
		Let $V$ be a finite set and let $L_2L_1: \mathscr{P}(V) \rightarrow \mathscr{P}(V)$ be a fully defined successive approximation operator on $\mathscr{P}(V).$ If $(S, R)$ is the solution returned by Algorithm 4.1,   is such that it is the unique solution then following holds:
		
		\vspace{2mm}
		\noindent  For any $x \in V$ we have that;\\
		(i)  $[x]_S \not\supseteq [x]_R$ unless $|[x]_R| =1$ \\
		(ii) $[x]_R \not\supseteq [x]_S$ unless $|[x]_S| =1,$  
	\end{corollary}

	\begin{proof}
		This follows directly from the conditions in Theorem \ref{t3}.
	\end{proof}
	
	\noindent \textbf{Example 4.1} (\emph{revisited}): Consider again, the given output vector of Example 4.1. First we form the $(S, R)$ pair using Algorithm 4.1. We get that $R = \{  \{a, b\}, \{c, d\}, \{e\}   \}$ and $S = \{  \{a, c\}, \{b\}, \{d, e\}  \}.$ Since this is the pair produced from Algorithm 4.1, we know that it satisfies the conditions of Theorem \ref{t2}. Now we need only to check if this pair satisfies the conditions of Theorem \ref{t3} to see if it is the only solution to do so. To keep track of which  equivalence class a set belongs to, we will index a set belonging to either $S$ or $R$ by $S$ or $R$ respectively. Then we see that $|\{a, b\}_R \cap \{b\}_S| = 1,$ $|\{c, d\}_R \cap \{a, c\}_S| = 1$ and $| \{e\}_R \cap \{d,e\}_S| = 1.$ This verifies both conditions of Theorem \ref{t3} and therefore this is the unique solution of the given operator.

	\begin{proposition} \label{p17}
		Let $V$ be a finite set and $L_2L_1 : \mathscr{P}(V) \rightarrow \mathscr{P}(V)$ be  a given defined operator. If $(E_1, E_2)$ is a unique solution such that either $E_1 \neq Id$ or $E_2 \neq Id$ where  $Id$ is the identity equivalence relation on $V$  then, 
		
		\vspace{2mm}
		\noindent (i) $  E_1 \not\leq E_2, $\\
		(ii) $ E_2 \not\leq E_1.$

	\end{proposition}
	
	\begin{proof}
		We first observe that if $E_1$ and $E_2$ are unique solutions and both of them are not $Id$ then one of them cannot be equal $Id.$ This is because if $(E_1, Id)$ were solutions to a given $L_2L_1$ operator corresponding to $L_1$ and $L_2$ respectively then $(Id, E_1)$ would also be solutions corresponding to $L_1$ and $L_2$ respectively and the solutions would not be unique. Hence, each of $E_1$ and $E_2$ contains at least one equivalence class of size greater than or equal to two.
		
		Suppose that $E_1 \leq E_2.$ Consider an $e \in E_2$ such that $|e| \geq 2.$ Then $e$ either contains a $f \in E_1$ such that $|f| \geq 2$ or two or more singletons in $E_1.$ Then first violates the condition of Corollary \ref{c2} and the second violates the second condition of Theorem \ref{t2}. Hence the solutions cannot be unique. Similarly, if we suppose that $E_2 \leq E_1.$   
	\end{proof} 
	
	\begin{corollary}
		Let $V$ be a finite set and $L_2L_1 : \mathscr{P}(V) \rightarrow \mathscr{P}(V)$ be a given defined operator. If there exists a unique solution $(E_1, E_2)$  such that either $E_1 \neq Id$ or $E_2 \neq Id$ where $Id$ is the identity equivalence relation on $V$  then,
		
		\vspace{2mm}
		\noindent (i) $ k = \gamma(E_1, E_2) = \frac {|POS_{E_1} (E_2)|}{|V|} < 1 $ or $E_1 \not\Rightarrow E_2$\\
		(ii) $ k = \gamma(E_2, E_1) = \frac {|POS_{E_2} (E_1)|}{|V|}  < 1$ or $E_2 \not\Rightarrow E_1.$
		
	\end{corollary}
	
	\begin{proof}
		This follows immediately from definitions.	
	\end{proof}
	
	\begin{proposition} \label{p16}
		Let $V$ be a finite set and $L_2L_1 : \mathscr{P}(V) \rightarrow \mathscr{P}(V)$ be a given defined operator. If there exists exists a unique solution $(E_1, E_2)$   then,
		
		\vspace{2mm}
		\noindent (i) for any $[x]_{E_1} \in E_1,$  $|POS_{E_2}([x]_{E_1})| \leq  1$	
		
		\noindent (ii) for any $[x]_{E_2} \in E_2,$  $|POS_{E_1}([x]_{E_2})| \leq  1.$ 
		
	\end{proposition} 
	
	\begin{proof} 
		This follows from the conditions in Theorem \ref{t4} and Corollary \ref{c2} which imply that for a unique pair solution $(E_1, E_2)$, an equivalence class of one of the equivalence relations  cannot contain any elements of size greater than one  of the other relation and can contain at most one element  of size exactly one of the other relation. 
	\end{proof}   
	
	\begin{corollary}
		Let $V$ be a finite set where $|V| = l$ and $L_2L_1 : \mathscr{P}(V) \rightarrow \mathscr{P}(V)$ be a given defined operator. If there exists exists a unique solution $(E_1, E_2)$   such that $ |E_1| = n$ and $|E_2| = m$ then,
		
		\vspace{2mm}
		\noindent (i) $ k = \gamma(E_1, E_2) = \frac {|POS_{E_1} (E_2)|}{|V|} \leq \frac{m}{l} $ \\
		(ii) $ k = \gamma(E_2, E_1) = \frac {|POS_{E_2} (E_1)|}{|V|} \leq \frac{n}{l} .$
		
		\begin{proof}
			Let $(E_1, E_2) $ be the unique solution of the given $L_2L_1$ operator. This result follows directly from the previous proposition by summing over all the elements in one member  of this pair for  taking its positive region	with respect to the other member of the pair.
		\end{proof}	
		
	\end{corollary}
	
	\begin{corollary}
		Let $V$ be a finite set such that $|V| = n$ and $L_2L_1 : \mathscr{P}(V) \rightarrow \mathscr{P}(V)$ be a given defined operator. If there exists exists a unique solution $(E_1, E_2)$  then,
		
		\vspace{2mm}
		\noindent (i) if the minimum size of an equivalence class in $E_1,$  $k_1$ where $	k_1 \geq 2$ then   \\ $ k = \gamma(E_1, E_2) = \frac {|POS_{E_1} (E_2)|}{|V|} = 0.$
		
		\vspace{2mm}
		\noindent (ii)  if the minimum size of an equivalence class in $E_2,$  $k_2$ where $	k_2 \geq 2$ then \\  $ k = \gamma(E_2, E_1) = \frac {|POS_{E_2} (E_1)|}{|V|} = 0.$
		
		\vspace{2mm}
	\end{corollary}

	\begin{proof}  
		Since no member of $E_2$ can contain any member of $E_1$ because $E_1$ has no singletons, we get that $\frac {|POS_{E_1} (E_2)|}{|V|} = 0.$  Similarly for Part (ii). 
	\end{proof}
	
	\begin{proposition}
		Let $V$ be a finite set  and $L_2L_1 : \mathscr{P}(V) \rightarrow \mathscr{P}(V)$ be a given defined operator. If there exists a unique solution $(E_1, E_2)$  such that $|E_1| = m$ and $|E_2| = n$ and $S_1$ is the number of singletons in $E_1$ and  $S_2$ is the number of singletons in $E_2,$ then,
		
		\vspace{2mm}
		\noindent (i) $ S_1 \leq n$
		
		\noindent (ii) $ S_2 \leq m.$
		
	\end{proposition}

	\begin{proof}
		We note that the conditions in Theorem \ref{t4}  imply that no two singletons in $E_1$ can be contained by any equivalence class in $E_2$ and vice versa. 	The result thus follows on application of the pigeonhole principle between the singletons in one equivalence relation and the number of elements in the other relation. 
	\end{proof}

	\subsection{A Derived Preclusive Relation and a Notion of  \\ Independence} 
	
	In \cite{PR}, Cattaneo and Ciucci found that preclusive relations are quite useful for using rough approximations in information systems.  In this direction,  we will  define a related notion of independence of equivalence relations from it. 
	
	Let $V$ be a finite set and let $\mathfrak{E}_V$ be the set of all equivalence relations on $V.$ Also, let $\mathfrak{E}_V^0 = \mathfrak{E}_V - Id_V,$ where $Id_V$ is the identity relation on $V.$  From now on, where the context is clear, we will omit the subscript. We now define a relation on $\mathfrak{E^0}, \ $ $\not\Rightarrow_{\mathfrak{E^0} },  $ as follows:   
	
	\vspace{2mm}  
	Let $E_1$ and $E_2$ be in $\mathfrak{E^0}.$ Let $L_2L_1: \mathscr{P}(V) \rightarrow \mathscr{P}(V)$ where $L_1$ and $L_2$ are lower approximation operators based on $E_1$ and $E_2$ respectively. Then,
	
	\vspace{2mm}
	\begin{center}
		$E_1 \not\Rightarrow_{\mathfrak{E^0}}  E_2$ iff $L_2L_1$ is a unique approximation operator. 
	\end{center}
	\vspace{2mm}
	
	\noindent That is, if for no other $E_3$ and $E_4$ in $\mathfrak{E^0}$ where at least one of $E_1 \neq E_3$ or $E_2 \neq E_4$ holds,  is it the case that the operator $L_2L_1 = L_3L_4,$ where $L_3$ and $L_4$ are lower approximation operators based on $E_3$ and $E_4$ respectively.
	
	\begin{definition}
		Let $V$ be a set and $E_1, E_2 \in 	\mathfrak{E}_V^0$. We say that $E_1$ is $ \mathfrak{E}_V^0$--\textbf{independent} of $E_2$ iff $E_1 \not\Rightarrow_{\mathfrak{E}_V^0}  E_2.$ Also, if $\lnot (E_1 \not\Rightarrow_{\mathfrak{E}_V^0}  E_2), $ we simply write $E_1 \Rightarrow_{\mathfrak{E}_V^0} E_2.$ Here, we say the  $E_1$ is $ \mathfrak{E}_V^0$--\textbf{dependent} of $E_2$ iff $E_1 \Rightarrow_{\mathfrak{E}_V^0} E_2.$ 
		
	\end{definition}
	
	\begin{proposition}
		$\not\Rightarrow_{\mathfrak{E}_V^0}$ is a preclusive relation.	 
	\end{proposition}
	
	\begin{proof} 
		We recall that a preclusive relation is one which is irreflexive and symmetric.	Let $E \in \mathfrak{E^0}_V.$ Since $E \neq Id,$ then by application of Proposition 4.2.3 $(E, E)$ does not generate a unique $L_2L_1$ operator and therefore $E \Rightarrow_{\mathfrak{E}_V^0} E.$ Hence $ \not\Rightarrow_{\mathfrak{E}_V^0}$ is irreflexive.
		
		Now, suppose that $E_1, E_2 \in \mathfrak{E^0}_V$ are such that $ E_1 \not\Rightarrow_{\mathfrak{E}_V^0} E_2.$ Then $(E_1, E_2)$ satisfies the conditions of Theorem \ref{t4}. Since together, the four conditions of the theorem are symmetric (with conditions (i) and (ii) and conditions (iii) and (iv) being symmetric pairs), then $(E_2, E_1)$ also satisfies the conditions of the theorem. Then by this theorem, we will have that $ E_2 \not\Rightarrow_{\mathfrak{E}_V^0} E_1.$ Hence, $\not\Rightarrow_{\mathfrak{E}_V^0}$ is symmetric.
	\end{proof} 
	
	\noindent \textbf{Remark 4.4} From the previous proposition we can see that dependency relation $\Rightarrow_{\mathfrak{E}_V^0}$ is a similarity relation.
	
	\begin{proposition}\label{p18}
		If $E_1 \Rightarrow E_2$ then $E_1 \Rightarrow_{\mathfrak{E}_V^0} E_2.$ 	
		
	\end{proposition} 
	
	\begin{proof} 
		This follows from Corollary \ref{p17}.   
	\end{proof}
	
	\begin{proposition} \label{p19}
		It is not the case that  $E_1 \Rightarrow_{\mathfrak{E}_V^0} E_2$ implies that $E_1 \Rightarrow E_2.$
	\end{proposition}	
	
	\begin{proof}
		In  Example 4.2 we see  $(E_1, E_2)$  does not give a corresponding unique $L_2L_1$ operator, hence  $E_1 \Rightarrow_{\mathfrak{E}_V^0} E_2$ but $E_1 \not\Rightarrow E_2.$ 
	\end{proof}  
	
	\noindent \textbf{Remark 4.5}  From Proposition \ref{p18} and Proposition \ref{p19},  we see that $  \mathfrak{E}_V^0$--\textit{\textbf{dependency}} is a more general notion of equivalence relation dependency that $\Rightarrow$ (or equivalently $\leq$ ). Similarly $ \mathfrak{E}_V^0$--\textit{\textbf{independence}}  is a stricter notion of independence than $\not\Rightarrow.$
	
	\begin{theorem} 
		Let $V$ be a finite set and $E_1$ and $E_2$ equivalence relations on $V.$ Then \\ $E_1 \not\Rightarrow_{\mathfrak{E}_V^0} E_2$ iff the following holds:
		
		\vspace{2mm}
		\noindent (i) For each $[x]_{E_2}, [y]_{E_2} \in E_2,$ if $[x]_{E_2} \neq [y]_{E_2}$ then $\textbf{u}_{E_1}([x]_{E_2}) \neq \textbf{u}_{E_1}([y]_{E_2}), $  
		
		\vspace{2mm}
		\noindent (ii) For each $[x]_{E_1}, [y]_{E_1} \in E_1,$ if $[x]_{E_1} \neq [y]_{E_1}$ then $\textbf{u}_{E_2}([x]_{E_1}) \neq \textbf{u}_{E_2}([y]_{E_1})$.
		
		\vspace{2mm}
		\noindent (iii) For any $[x]_{E_2} \in E_2,$ 
		there exists $[z]_{E_1} \in E_1$ such that, $ |[x]_{E_2} \cap [z]_{E_1}| = 1.$ 
		
		\vspace{2mm}
		\noindent(iv) For any $[x]_{E_1} \in E_1,$ 
		there exists $[z]_{E_2} \in E_2$ such that, $ |[x]_{E_1} \cap [z]_{E_2}| = 1.$

	\end{theorem} 
	
	\begin{proof}
		This follows directly from Theorem \ref{t4}.	
		
	\end{proof}

	\subsection{Seeing One Equivalence Relation through Another}
	
	We will first give a proposition which will show a more explicit symmetry between conditions (i) and (ii) and conditions (iii) and (iv) in Theorem \ref{t4} for unique solutions.
	
	\begin{proposition}
		Let $V$ be a finite set and let $E_1$ and $E_2$ be two equivalence relations on $V.$ Then;
		
		\vspace{2mm}
		\noindent  For any $[x]_{E_1} \in E_1,$ 
		$\exists[z]_{E_2} \in E_2$ such that, $ |[x]_{E_1} \cap [z]_{E_2}| = 1$ iff it is not the case that $ \exists Y, Z \in \mathscr{P}(V)$ such that $[x]_{E_1} = Y \cup Z,$ $Y \cap Z = \emptyset$ and $\textbf{u}_{E_2}(Y) = \textbf{u}_{E_2}(Z) = \textbf{u}_{E_2}([x]_{E_1}) .$ 
		
	\end{proposition}
	
	\begin{proof}
		We prove $\Rightarrow$ first.  Let $[x]_{E_1} \in E_1$ and suppose that $\exists[z]_{E_2} \in E_2$ such that, $ |[x]_{E_1} \cap [z]_{E_2}| = 1.$	Then let $[x]_{E_1} \cap [z]_{E_2} = t.$ Now for any spilt of $[x]_{E_1}$, that is for any $Y, Z \in \mathscr{P}(V)$ such that $[x]_{E_2} = Y \cup Z$ and $Y \cap Z = \emptyset,$ $t$ is in exactly one of these sets. Thus exactly one of $ \textit{\textbf{u}}_{E_2}(Y), \ \textit{\textbf{u}}_{E_2}(Z)$ contains $[t]_{E_2} = [z]_{E_2}.$ Hence $ \textit{\textbf{u}}_{E_2}(Y) \neq \textit{\textbf{u}}_{E_2}(Z).$

		We prove the converse by the contrapositive. Let $[x]_{E_1} \in E_1$ be such that for all $[z]_{E_2} \in E_2$ whenever $[x]_{E_1} \cap [z]_{E_2} \neq \emptyset$ (and clearly some such $[z]_{E_2}$ must exist), we have that $|[x]_{E_1} \cap [z]_{E_2}| \geq 2.$  Suppose that $[x]_{E_1}$ has non-empty intersection with with $n$ sets in $E_2.$ We note that $n \geq 1.$ Form a sequence of these sets; $R_1, \dots R_n. $ Since $|[x]_{E_1} \cap R_i| \geq 2$ for each $i$ such that $i= 1, \dots n,$ let $\{a_{i1}, a_{i2} \}$ be in $[x]_{E_1}\cap R_i$ for each $i$ such that $i= 1, \dots n.$ Let $Y = \{a_{i1}\ | \ i= 1, \dots n\}$ and let $Z = [x]_{E_1} - Y.$ Then, $[x]_{E_1} = Y \cup Z,$ $Y \cap Z = \emptyset$ and $\textit{\textbf{u}}_{E_2}(Y) = \textit{\textbf{u}}_{E_2}(Z) = \textit{\textbf{u}}_{E_2}([x]_{E_1}) .$ 
	\end{proof}
	
	\noindent Using the preceding proposition we obtain an equivalent form of Theorem \ref{t4}.
	
	\begin{theorem}
		Let $V$ be a finite set and $E_1$ and $E_2$ equivalence relations on $V.$ Then $(E_1, E_2)$ produces a unique $L_2L_1: \mathscr{P}(V) \rightarrow \mathscr{P}(V)$ operator  iff the following holds:

		\vspace{2mm}
		\noindent (i) For each $[x]_{E_2}, [y]_{E_2} \in E_2,$ if $[x]_{E_2} \neq [y]_{E_2}$ then $\textbf{u}_{E_1}([x]_{E_2}) \neq \textbf{u}_{E_1}([y]_{E_2})$  
		
		\vspace{2mm}
		\noindent (ii) For each $[x]_{E_1}, [y]_{E_1} \in E_1,$ if $[x]_{E_1} \neq [y]_{E_1}$ then $\textbf{u}_{E_2}([x]_{E_1}) \neq \textbf{u}_{E_2}([y]_{E_1})$
		
		\vspace{2mm}
		\noindent (iii) For any  $[x]_{E_2} \in E_2,$ if $ \exists Y, Z \in \mathscr{P}(V)$ such that $[x]_{E_2} = Y \cup Z$ and $Y \cap Z = \emptyset$ \\ \textcolor{white}{aaa} then $\textbf{u}_{E_1}(Y) \neq \textbf{u}_{E_1}(Z)$ 
		
		\vspace{2mm}
		\noindent(iv) For any  $[x]_{E_1} \in E_1,$ if  $ \exists Y, Z \in \mathscr{P}(V)$ if $[x]_{E_1} = Y \cup Z$ and $Y \cap Z = \emptyset$ then \\ \textcolor{white}{aaa}  $\textbf{u}_{E_2}(Y) \neq \textbf{u}_{E_2}(Z)$

	\end{theorem} 
	
	\subsubsection{Conceptual Translation of the Uniqueness Theorem}
	
	The  conditions of the above theorem can be viewed conceptually as follows:  (i) Through the eyes of $E_1,$ no two equivalence classes of $E_2$ are the same; (ii) Through the eyes of $E_2,$ no two equivalence classes of $E_1$ are the same;  (iii) No equivalence class in $E_2$ can be broken down into two smaller equivalence classes which are equal to it through the eyes of $E_1;$ (iv) No equivalence class in $E_1$ can be broken down into two smaller equivalence classes which are equal to it through the eyes of $E_2.$
	In other words we view set $V$\textbf{mod} $E_1.$ That is, let $V$\textbf{mod}$E_1$ be the set obtained from $V$ after  renaming the elements of $V$ with  fixed  representatives of their respective equivalence classes in $E_1.$  Similarly let $V$\textbf{mod}$E_2$ be the set obtained from $V$ after  renaming the elements of $V$ with  fixed  representatives of their respective equivalence classes in $E_2.$ We then have the following equivalent conceptual version of Theorem \ref{t4}

	\begin{theorem} 
		
		Let $V$ be a finite set and $E_1$ and $E_2$ equivalence relations on $V.$ Then \\ $(E_1, E_2)$  generate a unique $L_2L_1$ operator iff the following holds:

		\vspace{2mm}
		\noindent (i) No two distinct members of $E_2$ are equivalent in $V$\textbf{mod}$E_1.$
		
		\vspace{2mm}
		\noindent (ii) No two distinct members of $E_1$ are equivalent in $V$\textbf{mod}$E_2.$
		
		\vspace{2mm}
		\noindent (iii) No member $E_2$ can be broken down into two smaller sets which are equivalent to it in $V$\textbf{mod}$E_1.$

		\vspace{2mm}
		\noindent(iv) No member $E_1$ can be broken down into two smaller sets which are equivalent to it in $V$\textbf{mod}$E_2.$

	\end{theorem}

	\section{Decomposing $U_2U_1$ Approximations}
	
	We now investigate the case of double upper approximations. This is dually  related to the case of double lower approximations because of the relationship between upper and lower approximations by the equation, $U(X) = -L(-X)$ (see property 10 in Section 2.1.1). The following proposition shows that the problem of finding solutions for this case reduces to the case in the previous section:	\\   
	
	\begin{proposition} \label{p20}
		Let $V$ be a finite set and let $U_2U_1: \mathscr{P}(V) \rightarrow \mathscr{P}(V)$ be a given fully defined operator on $\mathscr{P}(V).$   Then  any solution $(E_1, E_2),$ is also a solution of $L_2L_1: \mathscr{P}(V) \rightarrow P\mathscr{P}(V)$ operator where  $L_2L_1(X) = -U_2U_1(-X)$ for any $X \in \mathscr{P}(V).$  Therefore, the solution $(E_1, E_2)$ for the defined $U_2U_1$ operator is a unique iff the solution for the corresponding $L_2L_1$ operator is unique.	 
		
	\end{proposition}
	
	\begin{proof}
		Recall that	$ L_2L_1(X) = -U_2U_1(-X)$. Hence, if there exists a solution $(E_1, E_2)$ which corresponds to the given $U_2U_1$ operator, this solution corresponds to a solution for the $L_2L_1$ operator  which is based on the same $(E_1, E_2)$ by the equation  $ L_2L_1(X) = -U_2U_1(-X).$  Similarly for the converse.
	\end{proof}
	
	\noindent \textbf{Algorithm:} 	Let $V$ be a finite set and let $U_2U_1: \mathscr{P}(V) \rightarrow \mathscr{P}(V)$ be a given fully defined operator on $\mathscr{P}(V).$  To solve for a solution, change it to solving for a solution for the corresponding $L_2L_1$ operator by the equation $ L_2L_1(X) = -U_2U_1(-X).$ Then, when we want to know the $L_2L_1$ output of a set we look at the $U_2U_1$ output of its complement set and take the complement of that. Next, use Algorithm 4.2 and the solution found will also be a solution for the initial $U_2U_1$ operator.

	\subsection{Characterising Unique Solutions}

	\begin{theorem}
		Let $V$ be a finite set and let $U_2U_1: \mathscr{P}(V) \rightarrow \mathscr{P}(V)$ be a given fully defined operator on $\mathscr{P}(V).$ If $(E_1, E_2)$ is a solution then, it is unique iff the following holds:

		\vspace{2mm}
		\noindent (i) for each $[x]_{E_2}, [y]_{E_2} \in E_2,$ if $[x]_{E_2} \neq [y]_{E_2}$ then $\textbf{u}_{E_1}([x]_{E_2}) \neq \textbf{u}_{E_1}([y]_{E_2}), $  
		
		\vspace{2mm}
		\noindent (ii) for each $[x]_{E_1}, [y]_{E_1} \in E_1,$ if $[x]_{E_1} \neq [y]]_{E_1}$ then $\textbf{u}_{E_2}([x]_{E_1}) \neq \textbf{u}_{E_2}([y]_{E_1})$.
		
		\vspace{2mm}
		\noindent (iii) For any $[x]_{E_2} \in E_2,$ 
		there exists $[z]_{E_1} \in E_1$ such that, $ |[x]_{E_2} \cap [z]_{E_1}| = 1.$ 
		
		\vspace{2mm}
		\noindent (iv) For any $[x]_{E_1} \in E_1,$ 
		there exists $[z]_{E_2} \in E_2$ such that, $ |[x]_{E_1} \cap [z]_{E_2}| = 1.$  
		
	\end{theorem}

	\begin{proof}
		This follows from Proposition \ref{p20} using Theorem \ref{t4}.
	\end{proof}

	\section{Decomposing $U_2L_1$ Approximations} 
	For this case, we observe that $U_2L_1 (X) = -L_2(-L_1(X)) = U_2(-U_1(-X)).$ Since we cannot get rid of the minus sign between the $L$s (or $U$s), duality will not save us the work of further proof here like it did in the previous section.  In this section, we will see that $U_2L_1$ approximations are tighter than $L_2L_1$ (or $U_2U_1$) approximations. For this decomposition  we will use an algorithm that is very similar to Algorithm 4.1, however notice the difference in step 2 where it only requires the use of minimal sets with respect to $\subseteq$ instead of minimum sets (which may not necessarily exist).\\ 
	
	\noindent \textbf{Algorithm 4.2: For Partial Decomposition of Double Successive Lower Approximations}
	
	\vspace{4mm}
	\noindent Let $V$ be a finite set. Given an input of a  fully defined operator $U_2L_1 : \mathscr{P}(V) \rightarrow \mathscr{P}(V),$  if a solution exists, we can produce a solution $(S, R)$, i.e. where $L_1$ and $U_2$ are the lower and upper approximation operators of equivalence relations $S$ and $R$ respectively,  by performing the following steps:
	
	\vspace{4mm}

	\vspace{3mm}
	\noindent \textbf{1}. Let $J$ be the set of output sets of the given $U_2L_1$ operator. We form the relation $R$ to be such that for $a, b \in V,$ $a \sim_R b \iff (a \in X \iff b\in X)$ for any $X \in J.$ It is clear that $R$ is an equivalence relation.

	\vspace{3mm}
	\noindent \textbf{2}. For each $Y \neq \emptyset$ output set, find the minimal pre-image sets with respect to $\subseteq,$ $Y_m,$ such that $U_2L_1(Y_m) = Y$. Collect all these minimal sets in a set $K.$  Note that we can always find these minimal sets since $\mathscr{P}(V)$ is finite.
	
	\vspace{3mm}
	\noindent \textbf{3}. Using $K,$ we form the relation $S$ to be such that for $a, b \in V,$ $a \sim_S b \iff (a \in X \iff b\in X)$ for any $X \in K.$ It is clear that $S$ is an equivalence relation.
	
	\vspace{3mm}
	\noindent \textbf{4}. Form the operator $U_RL_S : \mathscr{P}(V) \rightarrow \mathscr{P}(V)$ generated by $(S, R).$ If for all $X \in \mathscr{P}(V)$, the given $U_2L_1$ operator is such that $U_2L_1(X) = U_RL_S(X),$ then $(S, R)$ is a solution proving that a solution exists (note that it is not necessarily unique). Return $(S, R).$ Otherwise, discard $S$ and $R$ and return 0 signifying that no solution exists.\\
	
	\noindent We will  prove the claim in step 4 in this section. \\

	\begin{lemma} \label{p21}
		Let $V$ be a set and $U_2L_1 : \mathscr{P}(V) \rightarrow \mathscr{P}(V)$ be a given fully defined operator on $\mathscr{P}(V)$ with $L_1$ and $E_2$ based on unknown $E_1$ and $E_2$ respectively.	 Let $R$ and $S$ be equivalence relations defined on $V$ as constructed in Algorithm 4.3. Then  $E_2\leq R$ and $E_1 = S.$
		
	\end{lemma}
	
	\begin{proof}
		We first prove 	$E_2\leq R.$ Now the output set of a non-empty set in $\mathscr{P}(V)$ is obtained by first applying the lower approximation $L_1$ to it and and after applying the upper approximation, $U_2$ to it. Hence by definition of $U_2,$ the non-empty output sets are unions of equivalence classes of the equivalence relation which corresponds to $U_2.$   If $a$ is in an output set but $b$ is not in it then they cannot belong to the same equivalence class of  $E_2$ i.e. $a \not\sim_R b$ implies that $a \not\sim_{E_2} b.$ Hence $E_2\leq R. $   
		
		Now, the minimal pre-image, X say,  of a non-empty output set which is a  union of equivalence classes in $E_2,$ has to be a union of equivalence classes in $E_1.$ For suppose it was not. Let $Y = \{y \in X\ | \ [y]_{E_1} \not\subseteq X\}.$ By assumption, $Y \neq \emptyset.$ Then $L_1 (X) = L_1 (X - Y).$ Hence $U_2L_1(X) = U_2L_1 (X - Y)$ but $|X - Y| < |X|$ contradicting minimality of $X$. Therefore, if $a$ belongs to the minimal pre-image of a non-empty output set but $b$ does not belong to it, then $a$ and $b$ cannot belong to the same equivalence class in $E_1$ i.e. $a \not\sim_S b$ which implies that $a \not\sim_{E_1} b.$ Hence $E_1\leq S.$  
		
		We now prove the converse, that $ S \leq E_1.$ For suppose it was not. That is, $E_1 < S.$ Then there exists at least one equivalence class in $S$ which is split into smaller equivalence classes in $E_1.$ Call this equivalence class $[a]_S.$ Then there exists $w, t \in V$  such that $[w]_{E_1} \subset [a]_S$ and $[t]_{E_1} \subset[a]_S.$ Now consider the pre-images of a  minimal output sets of $U_2L_1,$ containing $t.$ That is, $X$ such that $U_2L_1(X) = Y$ where $Y$ is the minimal output set such that $t \in Y$ and for any $X_1 \subset X,$  $U_2L_1(X_1) \neq Y.$  The following is a very useful observation.
		
		\vspace{2mm}
		\noindent \textbf{Claim:} For any $v \in \textbf{\textit{u}}_{E_1}([y]_{E_2}),$ $[v]_S$ is a minimal set such that  $U_2L_1([v]_S) \supseteq [y]_{E_2}.$  
		\vspace{2mm}
		\noindent The above follows because 1) $U_2L_1([v]_S) \supseteq [y]_{E_2}$  since $v \in \textbf{\textit{u}}_{E_1}([y]_{E_2})$ and 2) For any $Z \subset [v]_S, \ U_2L_1(Z) = \emptyset$ since $L_1(Z) = \emptyset.$ 
		
		Now for $U_2L_1(X)$ to contain $t,$ then it must contain $[t]_{E_2}.$  Hence by the previous claim, $X = [t]_S$ is such a  minimal  pre-image of a set containing $t$. If $L_1$ is based on $S,$ then $X = [t]_S = [a]_S.$  However, if $L_1$ is based on $E_1,$ then $ X = [a]_S$ is not such a minimal set because $X = [t]_{E_1}$ is such that $U_2L_1(X) = Y$ but $[t]_{E_1} \subset [a]_S.$ Hence,  $U_RL_S(X)\neq U_{E_2}L_{E_1}(X)$ for all $X \in \mathscr{P}(V)$ which is a contradiction to $(E_1, E_2)$ also being a solution for the given $U_2U_1$ operator. Thus we have that $E_1 = S.$     
	\end{proof}

	\begin{lemma} \label{l6}
		Let $V$ be a finite set and  $U_2L_1 : \mathscr{P}(V) \rightarrow \mathscr{P}(V)$ be a fully defined operator. If there exists  equivalence pair solutions to the operator $(E_1, E_2)$ which is such that there exists  $[x]_{E_2}, [y]_{E_2} \in E_2,$ such that $[x]_{E_2} \neq [y]_{E_2}$ and $\textbf{u}_{E_1}([x]_{E_2}) = \textbf{u}_{E_1}([y]_{E_2}), $ then there exists another solution, $(E_1, H_2)$, where $H_2$ is an equivalence relation formed from $E_2$ by combining $[x]_{E_2}$ and $[y]_{E_2}$  and all other elements are as in $E_2.$ That is, $[x]_{E_2} \cup [y]_{E_2} = [z] \in H_2$ and if $[w] \in E_2$ such that $[w] \neq [x]_{E_2}$ and $[w]_{E_2} \neq [y]_{E_2},$ then $[w] \in H_2.$
		
	\end{lemma}
	
	\begin{proof}
		Suppose that $(E_1, E_2)$ is a solution of a given $U_2L_1$ operator and $H_2$ is as defined above. Now, $U_2L_1(X) = Y $ iff the union of $E_1$-equivalence classes in $X$ intersects the equivalence classes of $E_2$ whose union is equal to $Y.$ So, in the $(E_1, H_2)$ solution, the only way that $U_{H_2}L_{E_1}(X)$ could be different from $U_{E_2}L_{E_1}(X)$(which is $=U_2L_1(X)$) is if there some equivalence class of $E_1$ which either intersects $[x]_{E_2}$ but not $[y]_{E_2}$ or intersects $[y]_{E_2}$ but not $[x]_{E_2}.$ However, this is not the case since we have that $\textit{\textbf{u}}_{E_1}([x]_{E_2}) = \textit{\textbf{u}}_{E_1}([y]_{E_2}).$ Hence, $U_{E_2}L_{E_1} (X) = U_{H_2}L_{E_1}(X)$ for all $X \in \mathscr{P}(V)$ and therefore if $(E_1, E_2)$ is a solution to the given operator then so is $(E_1, H_2).$
	\end{proof}

	\noindent Next, we prove the claim in step 4 of Algorithm 4.2.

	\begin{theorem} \label{t5}
		Let $V$ be a finite set and $U_2L_1 : \mathscr{P}(V) \rightarrow \mathscr{P}(V)$ a fully defined operator. If there exists an equivalence relation pair solution, then there exists a solution $(E_1, E_2),$  which satisfies, 
		
		\vspace{2mm}
		\noindent (i) for each $[x]_{E_2}, [y]_{E_2} \in E_2,$ if $[x]_{E_2} \neq [y]_{E_2}$ then $\textbf{u}_{E_1}([x]_{E_2}) \neq \textbf{u}_{E_1}([y]_{E_2}), $  
		
		\vspace{2mm}
		\noindent Furthermore $E_1 =S$ and $E_2 = R,$ where $(S, R)$ are the relations obtained by applying Algorithm 4.2 to the given $U_2L_1$ operator.
	\end{theorem}

	\begin{proof} 
		
		Suppose that there exists a solution $(C, D).$ Then by Lemma \ref{p21}, $C = S,$ where $S$ is produced by Algorithm 4.2. If $(S, D)$ satisfies condition (i) of the theorem then take $(E_1, E_2) = (C, D).$ Otherwise, use repeated applications of Lemma \ref{l6} until we obtain a solution, $(S, E_2)$ which satisfies the condition of the theorem. Since $\mathscr{P}(V)$ is finite this occurs after a finite number of applications of the lemma. Moreover, by Lemma \ref{p21}, $E_2 \leq R.$

		Consider the minimal sets in the output list of the given $U_2L_1$ operator. It is clear that these sets are union of one or more equivalence classes of $E_2.$ Let $[y]_{E_2} \in E_2$ then for any $v \in \textbf{\textit{u}}_{E_1}([y]_{E_2})),$  $U_2L_1([v]_S) \supseteq [y]_{E_2}$ (by the claim in Lemma \ref{p21}). 
		
		\vspace{2mm}	
		\noindent \textbf{Claim:} (i) For any $[y]_{E_2} \neq [z]_{E_2} \in E_2,$ there exists an output set, $U_2L_1(X)$ such that it contains at least of $[y]_{E_2}$ or $[z]_{E_2}$ both it does not contain both sets.
		
		\vspace{2mm}
		Suppose that $[y]_{E_2} \neq [z]_{E_2} \in E_2.$ By the assumed condition of the theorem, then $\textbf{\textit{u}}_{E_1}([y]_{E_2}) \neq \textbf{\textit{u}}_{E_1}([z]_{E_2}).$ Hence either (i) there exists $a \in V$ such that $a \in \textbf{\textit{u}}_{E_1}([y]_{E_2})$ and $a \not\in \textbf{\textit{u}}_{E_1}([z]_{E_2})$ or (ii) there exists $a \in V$ such that $a \not\in \textbf{\textit{u}}_{E_1}([y]_{E_2})$ and
		$a \in \textbf{\textit{u}}_{E_1}([z]_{E_2}).$ Consider the first case.  This implies that $[a]_S \cap [y]_{E_2} \neq \emptyset$  while $[a]_S \cap [z]_{E_2} = \emptyset.$ 
		Therefore, $U_2L_1([a]_S) \supseteq [y]_{E_2}$ but $U_2L_1([a]_S) \not\supseteq [z]_{E_2}.$ Similarly, for the second case we will get that  $U_2L_1([a]_S) \supseteq [z]_{E_2}$ but $U_2L_1([a]_S) \not\supseteq [y]_{E_2}$ and the claim is shown.
		
		We recall that  $a \sim_R b \iff \   (a \in X \iff b \in X)$ for each $X$ in the range of the given $U_2L_1$. From the previous proposition we have that $E_2 \leq R.$ From the above  claim we see that if $[y]_{E_2} \neq [z]_{E_2}$ in $E_2$ then there is an output set that contains one of $[y]_{E_2}$ or $[z]_{E_2},$ but not the other. Hence, if $x \not\sim_{E_2} y$ then $x \not\sim_R y.$ That is, $R \leq E_2.$ Therefore we have that $R = E_2.$
	\end{proof}

	\subsection{Characterising Unique Solutions}
	
	\begin{theorem} \label{t6}
		Let $V$ be a finite set and let $U_2L_1: \mathscr{P}(V) \rightarrow \mathscr{P}(V)$ be a fully defined successive approximation operator on $\mathscr{P}(V).$ If $(S, R)$ is returned by Algorithm 4.1,  then $(S, R)$ is the unique solution of the operator iff the following holds:
		
		\vspace{2mm}
		\noindent (i) For any $[x]_R \in R,$ 
		there exists $[z]_S \in S$ such that, $ |[x]_R \cap [z]_S| = 1.$ 
		
		\vspace{2mm}

	\end{theorem}	
	
	\begin{proof}
		We prove $\Leftarrow$  direction first. So assume the condition holds. Then by Theorem \ref{t5} if there is a unique solution, it is $(S, R)$ produced by Algorithm 4.2. We note that by Lemma \ref{p21}, any  other solution, $(E_1, E_2)$ to the given $U_2L_1$ operator   must be such that $E_1 = S$ and $E_2 \leq R.$

		So, suppose to get a contradiction, that there exists a solution $(E_1, E_2)$ which is such  that $E_2 < R.$ That is, $E_2$ contains a splitting of at least one of the equivalences classes of $R,$ say $[a]_R.$ Hence $|[a]_R| \geq 2.$ By assumption there exists a $[z]_S \in S$  such that $|[a]_R \cap [z]_S| = 1.$ Call the element in this intersection $v$ say. We note that $[v]_S= [z]_S.$ Now as $[a]_R$ is spilt into smaller classes in $E_2,$ $v$ must be in one of these classes, $[v]_{E_2}.$  Now, $U_2L_1([v]_S)$ when $U_2$ is based on $E_2,$ contains $[v] _{E_2}$ but does not contain $[a]_R.$ This is because $[v]_S \cap ([a]_R - [v]_{E_2}) = \emptyset.$ That is, $U_{E_2}L_S([v]_S) \not\supseteq [a]_R$ but $U_RL_S([v]_S) \supseteq [a]_R.$ Hence $U_{E_2}L_S (X) \neq U_RL_S(X)$ for all $X \in \mathscr{P}(V).$ This is a contradiction to $(S, E_2)$ also being a solution to the given $U_2L_1$ operator for which $(S, R)$ is a solution.  Hence we have a contradiction and so $E_2 = R.$

		Now we prove $\Rightarrow $ direction. Suppose that $(E_1, E_2)$ is the unique solution, and assume that the condition does not hold.  By uniqueness,  $(E_1, E_2) = (S, R).$ Then, there exists an $[x]_R \in R$ such that for all $[y]_S \in S$ such that $[x]_R \cap [y]_S \neq \emptyset$  we have that $|[x]_R \cap [y] _S| \geq 2.$ 
		
		Suppose that $[x]_R$ has non-empty intersection with with $n$ sets in $S.$ We note that $n \geq 1.$ Form a sequence of these sets; $S_1, \dots S_n. $ Since $|[x]_R \cap S_i| \geq 2$ for each $i$ such that $i= 1, \dots n,$ let $\{a_{i1}, a_{i2} \}$ be in $[x]_R \cap S_i$ for each $i$ such that $i= 1, \dots n.$ We split $[x]_R$ to form a finer $E_2$ as follows: Let $P = \{a_{i1}\ | \ i= 1, \dots n\}$ and $Q = [x]_R - P$ be two equivalence classes in $E_2$ and for the rest of $E_2,$ for any $[y]_R \in R$ such that $[y]_R \neq [x]_R, $ let $[y] \in E_2$ iff $[y] \in R.$ Now, $U_RL_S(X) = Y $ iff the union of $S$-equivalence classes in $X$ intersects equivalence classes of $E_2$ whose union is equal to $Y.$ So, for the $(S, E_2)$ solution, the only way that $L_{E_2}L_S(X)$ could be different from $L_RL_S(X)$  is if there is an equivalence class in $S$ which intersects $P$ but not $Q$ or $Q$ but not $P.$ However, this is not the case because $\textit{\textbf{u}}_S(P) = \textit{\textbf{u}}_S(Q).$ Hence, $L_RL_S(X) = L_{E_2}L_S(X)$ for all $X \in \mathscr{P}(V)$ and  if $(S, R)$ is a solution for the given vector, then so is $(S, E_2)$ which is a contradiction of assumed uniqueness of $(S, R).$
	\end{proof}  
	
	\noindent The following result sums up the effects of Theorem \ref{t5} and Theorem \ref{t6}. 
	
	\vspace{2mm}
	\begin{theorem}
		Let $V$ be a finite set and let $U_2L_1: \mathscr{P}(V) \rightarrow \mathscr{P}(V)$ be a given fully defined operator on $\mathscr{P}(V).$ Then there exists a unique pair of equivalence relations solution $(E_1, E_2)$  iff the following holds:

		\vspace{2mm}
		\noindent (i) for each $[x]_{E_2}, [y]_{E_2} \in E_2,$ if $[x]_{E_2} \neq [y]_{E_2}$ then $\textbf{u}_{E_1}([x]_{E_2}) \neq \textbf{u}_{E_1}([y]_{E_2}), $

		\vspace{2mm}
		\noindent (iii) For any $[x]_{E_2} \in E_2,$ 
		there exists $[z]_{E_1} \in E_1$ such that, $ |[x]_{E_2} \cap [z]_{E_1}| = 1.$

	\end{theorem}	
	
	\section{Decomposing $L_2U_1$ Approximations} 
	
	For this case we observe that $L_2U_1$ is dual to the case previously investigated $U_2L_1$ operator. Due to the duality connection between $L_2U_1$ and $U_2L_1$, the question of unique solutions of the former reduces to the latter as the following proposition shows. \\
	
	\begin{proposition} \label{p22}
		Let $V$ be a finite set and let $L_2U_1: \mathscr{P}(V) \rightarrow \mathscr{P}(V)$ be a given fully defined operator on $\mathscr{P}(V).$   Then  any solution $(E_1, E_2),$ is also a solution of $U_2L_1: \mathscr{P}(V) \rightarrow P\mathscr{P}(V)$ operator where  $U_2L_1(X) = -L_2U_1(-X)$ for any $X \in \mathscr{P}(V).$  Therefore, the solution $(E_1, E_2)$ for the defined $U_2U_1$ operator is a unique iff the solution for the corresponding $U_2L_1$ operator is unique.	 
		
	\end{proposition}  
	
	\begin{proof}
		Recall that	$ U_2L_1(X) = -L_2U_1(-X)$. Hence, if there exists a solution $(E_1, E_2)$ which corresponds to the given $U_2L_1$ operator, this solution corresponds to a solution for the $L_2U_1$ operator  which is based on the same $(E_1, E_2)$ by the equation  $ L_2U_1(X) = -U_2L_1(-X).$ Similarly for the converse. 
	\end{proof}
	
	\noindent \textbf{Algorithm:} 	Let $V$ be a finite set and let $L_2U_1: \mathscr{P}(V) \rightarrow \mathscr{P}(V)$ be a given fully defined operator on $\mathscr{P}(V).$  To solve for a  solution, change it to solving for a solution for the corresponding $U_2L_1$ operator by the equation $ U_2L_1(X) = -L_2U_1(-X).$ Then, when we want to know the $U_2L_1$ output of a set we look at the $L_2U_1$ output of its complement set and take the complement of that. Next, use Algorithm 4.2 and the solution found will also be a solution for the initial $L_2U_1$ operator.

	\subsection{Characterising Unique Solutions}

	\begin{theorem}
		
		Let $V$ be a finite set and let $L_2U_1: \mathscr{P}(V) \rightarrow \mathscr{P}(V)$ be a given fully defined operator on $\mathscr{P}(V).$ If $(E_1, E_2)$ is a solution, then it is unique iff the following holds:

		\vspace{2mm}
		\noindent (i) for each $[x]_{E_2}, [y]_{E_2} \in E_2,$ if $[x]_{E_2} \neq [y]_{E_2}$ then $\textbf{u}_{E_1}([x]_{E_2}) \neq \textbf{u}_{E_1}([y]_{E_2}), $ 
		
		\noindent (ii) For any $[x]_{E_2} \in E_2,$ 
		there exists $[z]_{E_1} \in E_1$ such that, $ |[x]_{E_2} \cap [z]_{E_1}| = 1.$

		\vspace{2mm} 
		
	\end{theorem}
	
	\begin{proof}
		This follows from Proposition \ref{p22}, Theorem \ref{t5} and Theorem \ref{t6}. 
	\end{proof}  

	\section{Conclusion}  
	We have defined and examined the consequences of double successive rough set approximations based on two, generally unequal equivalence relations on a finite set. We have given algorithms to decompose a given defined operator into constituent parts. Additionally, in sections 4.2 and 4.3 we have found a conceptual translation of the main results which is very much in the spirit of what Yao suggested in \cite{TSid}.  These type of links are especially helpful in forming a coherent map of the mass of existing literature out there. 
	
	This type of analysis can be seen as somewhat  analogous to decomposing a wave into constituent sine and cosine waves using Fourier analysis. In our case, we work out the possibilities of what can be reconstructed if we know that a system has in-built layered approximations. It is possible that some heuristics of how the brain works can be modelled using such approximations and cognitive science is a possible application for the theory which we have begun to work out.

	\bibliographystyle{plain}
	
	\bibliography{mybib3b}

\end{document}